\documentclass[runningheads]{llncs}
\usepackage{enumitem}
\usepackage{makeidx}
\usepackage{amsmath,amsfonts,amssymb}
\usepackage[T1]{fontenc} 
\usepackage[utf8]{inputenc}
\usepackage{graphicx}
\usepackage{hyperref}
\usepackage{tikz, pgf}
\usetikzlibrary{arrows,automata,backgrounds}
\usepackage{cancel}
\usepackage{subfigure}
\usepackage{algorithm}
\usepackage{algpseudocode}
\usepackage{algorithmicx}
    \algrenewcommand\algorithmicrequire{\textbf{Input:}}
    \algrenewcommand\algorithmicensure{\textbf{Output:}}

\newcommand{\pf}[1]{\text{Pref}(\mathcal #1)}

\newcommand{\mt}[1]{\mathcal #1} 
\newcommand{\ndashv}{\;\cancel\dashv\;}
\newcommand{\la}[1]{\mathcal L(\mathcal #1)}

\definecolor{cgreen}{rgb}{0.0, .45, 0.0}

\renewenvironment{proof}[1][\proofname]{{\noindent\bfseries #1. }}{\qed\vskip2mm}

\begin{document}
\title{Ordering Regular Languages and Automata: Complexity}
%
\author{Giovanna D'Agostino \and
Davide Martincigh \and
Alberto Policriti}
\authorrunning{G. D'Agostino et al.}
%
\institute{University of Udine, Italy}
\maketitle              

\begin{abstract}
Given an order of the  underlying alphabet we can lift it to the states of a finite deterministic 
automaton: to compare states we use the order of the strings reaching them. When the order on 
strings is the co-lexicographic one \emph{and} this order turns out to be total, the DFA is called 
Wheeler. This recently introduced class of automata---the \emph{Wheeler automata}---constitute an 
important data-structure for languages, since it allows the design and implementation of a very 
efficient tool-set of storage mechanisms for the transition function, supporting a large variety of 
substring queries.

In this context it is natural to consider the class of regular languages accepted by Wheeler automata, i.e. the 
Wheeler languages. An inspiring result in this area is the following: it has been shown that, as opposed to the general case, the classic determinization  by powerset construction is \emph{polynomial} on Wheeler automata. As a consequence, most  classical 
problems, when considered on this class of automata,  turn out to be ``easy''---that is, solvable in polynomial time.

In this paper we consider  computational problems related to Wheelerness, but starting from 
non-deterministic automata. We also consider the case of \emph{reduced} non-deterministic ones---a class of NFA where recognizing Wheelerness is still polynomial, as for DFA's. Our collection of results shows that moving towards non-determinism is, in most cases, a dangerous path leading quickly to 
intractability.

Moreover, we start a study of  ``state complexity'' related to Wheeler DFA and languages, proving that
the classic construction for the intersection of languages turns out to be computationally simpler on Wheeler DFA than in the general case. We also provide a construction for the minimum Wheeler DFA recognizing a given Wheeler language. 

\keywords{Regular languages \and Finite Automata \and Wheeler Automata \and Ordering Languages.}
\end{abstract}

\section{Introduction}

A simple and natural way of efficiently storing and composing regular languages presented by their accepting automata is by exploiting some kind of order imposed on their collection of states.  After all, ordering a collection of objects is very often a way to shed light on their internal structure and ease their manipulation. 

One  way of  ordering the states of a finite automaton is to consider their incoming languages---that is, the set of strings reaching the given states---and proposing a way to compare them.  If we fix an order on the underlying alphabet $\Sigma$  and  consider  states as ending points of strings, we are naturally invited to start from their \emph{last} character (the final one on the path reaching the state) and proceed backwards. This results in using the so-called \emph{co-lexicographic} order over $\Sigma^*$. 
Since incoming languages of  different states  of a deterministic automaton $\mt D$ do  not intersect, the co-lexicographic order can easily be lifted to the states of  $\mt D$:  $q\leq_{\mt D} q'$ if all strings of the incoming language of $q$ are co-lexicographically smaller than any string of the incoming language of $q'$.
This order turns out to be very useful, allowing  to store $\mathcal D$ using a \emph{succint index}, that is, a space-saving data structure that supports fast matching queries \cite{NN}. It turns out that the complexity of constructing such an index      depends 
on the {\em width} of the order $\leq_{\mt D}$ (see \cite{NN}), the best possible case being the one where  $\leq_{\mt D}$ is a total order.
In the latter case $\mt D$ is called a {\em Wheeler}  automaton, and in \cite{ADPP} it has been proved that recognizing Wheelerness is an easy task over DFA's.

When moving from DFA's to NFA's things become  more complicate and two possible approaches were considered: 
\begin{itemize}
    \item The first one  consists in identifying some  {\em local} properties of $\leq_\mt D$, used to define a  general notion of a    {\em co-lex} (possibly partial) order over the states of an NFA (see \cite{ADPP2}). Turning back to  DFA's, one can  easily prove that  $\leq_{\mt D}$ is the maximum co-lex (partial)  order over  $\mt D$. In general, co-lex orders over NFA's can  still   be used  for  indexing, with index-construction complexity parametric on the \emph{width} (i.e. the maximum length of an anti-chain in $\leq_\mt D$)  of the co-lex order. Unfortunately, co-lex orders are not as well behaving on NFA's as they are on  DFA's: over an NFA we cannot guarantee  the existence of a  maximum co-lex order and also finding a maximal one turns out to be an  NP-complete problem \cite{gibney2020simple}. To overcome such difficulties, in \cite{ADPP2}  a new class of automata was introduced: the {\em reduced NFA's}. On reduced NFA's  distinguished states have different incoming languages.  While allowing non-determinism, the reduced NFA's  share    with DFA's the good behaviour of co-lex orders: any reduced NFA possesses a polynomial time computable, maximum co-lex order, so that recognizing Wheelerness is no longer an NP-complete problem over them. 
   \item The second approach  consists in  generalizing the definition of $\leq_{\mt D}$ over NFA's states, by defining an order depending directly on the incoming languages. Such generalization must now take care of the fact that incoming languages may intersect. 
Actually, since in an  NFA $\mt A$ there could be    different states with the same incoming language, when lifting the order to the state of $\mt A$ we must be careful not to identify states with the same incoming languages.
\end{itemize}

As far as the first approach is concerned, in this paper we prove that deciding whether a language is Wheeler, i.e. whether it is recognized by an NFA with a total co-lex order, is PSPACE-complete. This remains the case even if we restrict to reduced NFA's. Note that   the same problem using a recognizing DFA was proved to be easy (polynomially computable) in \cite{ADPP}.

Regarding the second approach, even though the proposed partial order was shown to be useful for indexing \cite{JACM},  we first need to compute it. In this paper we prove that the
  task of computing $\leq_{\mt A}$ is  difficult over NFA's, even on the class of reduced NFA's. Actually, as a corollary of this fact we also see that recognizing reduced-ness is a difficult task.  The proof relies on the fact that the universality problem is PSPACE-complete over reduced NFA (as for the whole class of non-deterministic automata). 
  
 In the last part of the paper we go back to DFA's and tackle the problem of establishing the state complexity of the intersection of two Wheeler automata. We prove that equipping the input automata with an order on their collection of states allows us to do much better than in the general case: the standard procedure now turns out of a complexity proportional to the sum the sizes of the input automata. 
 Our final result regards the (difficult) problem of computing  the minimum-size Wheeler automaton, starting from the minimum automaton accepting a given Wheeler language.

\subsection{Preliminaries}

First of all, we fix some notation. 
Given a total order $(Z,<)$ we say that a subset $I\subseteq Z$ is an \emph{interval} if, for any $x,y,z\in Z$ with $x<y<z$,  if $x,z\in I$ then $y\in I$.
Let $\Sigma$ denote a finite alphabet endowed with a total order $(\Sigma,\prec)$. We denote by $\Sigma^*$ the set of finite strings over $\Sigma$, with $\varepsilon$ being the empty string. We extend the order $\prec$ over $\Sigma$ to the \emph{co-lexicographic} order $(\Sigma^*, \prec)$, where $\alpha \prec \beta$ if and only the reverse of $\alpha$, i.e. $\alpha$ read from the right to the left, precedes lexicographically the reverse of $\beta$. Given two strings $\alpha, \beta \in \Sigma^*$, we denote by $\alpha \dashv \beta$ the property that $\alpha$ is a suffix of $\beta$. 
For a language $\mathcal L \subseteq \Sigma^*$, we denote by $\pf L$ the set of prefixes of strings in $\mathcal L$.
We denote by $\mathcal A = (Q, q_0, \delta, F, \Sigma)$ a finite automaton (NFA), with $Q$ as set of states, $q_0$ initial state, $\delta: Q \times \Sigma \rightarrow 2^Q$ transition function, and $F \subseteq Q$ final states. 
The size of $\mathcal A$, denoted by $|\mathcal A|$, is defined to be $|Q|$.
An automaton is deterministic (DFA) if $|\delta(q, a)| \le 1$, for all $q\in Q$ and $a\in \Sigma$. As customary, we extend $\delta$ to operate on strings as follows: for all $q\in Q$, $a\in \Sigma$ and $\alpha \in \Sigma^*$
\[
\delta(q,\varepsilon) = \{q\}, \qquad \delta(q,\alpha a)=\bigcup_{v\in \delta(q,\alpha)} \delta(v,a). 
\]
We denote by $\la A = \{\alpha \in \Sigma^*:\, \delta(q_0,\alpha) \cap F \ne \emptyset\}$ the language accepted by the automaton $\mathcal A$.
We assume that every automaton is \emph{trimmed}, that is, every state is reachable from the initial state and every state can reach at least one final state. Note that this assumption is not restrictive, since removing every state not reachable from $q_0$ and every state from which is impossible to reach a final state from an NFA, can be done in linear time and does not change the accepted language.
It immediately follows that: 
\begin{itemize}
    \item there might be only one state without incoming edges, namely $q_0$; 
    \item every string that can be read starting from $q_0$ belongs to $\pf L$.
\end{itemize}
We will often make use of the notion of the \emph{incoming language} of a state of an NFA, defined as follows.

\begin{definition}[Incoming language]
 Let $\mathcal A=(Q,q_0,\delta, F,\Sigma)$ be an NFA and let $q\in Q$. The \emph{incoming language} of $q$, denoted by $I_q$, is the set of strings that can be read on $\mathcal A$ starting from $q_0$ and ending in $q$. In other words, $I_q$ is the language recognized by the automaton $\mathcal A_q=(Q,q_0,\delta, \{q\},\Sigma)$.
\end{definition}


The class of Wheeler automata has been recently introduced in \cite{Gagie}. An automaton in this class has the property that there exists a total order on its states that is propagated along equally labeled transition. Moreover, the order must be compatible with the underlying order of the alphabet:

\begin{definition}[Wheeler Automaton]
\label{WheelerAutomaton}
A Wheeler NFA (WNFA) $\mathcal{A}$ is an NFA $(Q,q_0,\delta,F,\Sigma)$  endowed with a binary relation
<, such that: $(Q,<)$ is a linear order having the initial state $q_0$ as minimum, $q_0$ has no in-going edges, and
the following two (Wheeler) properties are satisfied. Let $v_1 \in \delta(u_1, a_1)$ and $v_2 \in \delta(u_2, a_2)$:
\begin{enumerate}[label = (\roman*)]
    \item $a_1 \prec a_2 \,\rightarrow \, v_1 < v_2$ 
    \item $(a_1 = a_2 \wedge u_1 < u_2) \,\rightarrow \, v_1 \le v_2$.
\end{enumerate}
A Wheeler DFA (WDFA) is a  deterministic WNFA.
\end{definition}

\begin{remark}
A consequence of Wheeler property (i) is that $\mt A$ is \emph{input-consistent}, that is all transitions
entering a given state $u \in Q$ have the same label: if $u \in \delta(v,a)$ and $u \in \delta(w,b)$, then $a=b$. Therefore the function $\lambda: Q\setminus \{q_0\} \rightarrow \Sigma$ that associate to each state the unique label of its incoming edges is well defined. For the state $q_0$, the only one without incoming edges, we set $\lambda(q_0) := \#$. 
\end{remark}

In Figure \ref{w} is depicted an example of a WDFA.

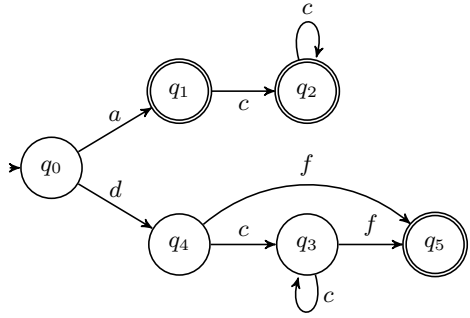
\begin{figure}[ht]%
\begin{center}
\begin{tikzpicture}[->,>=stealth', semithick, initial text={}, auto, scale=.34]
 \node[state, label=above:{}, initial] (0) at (0,0) {$q_0$};
 \node[state, label=above:{}, accepting] (1) at (5,3) {$q_1$};
 \node[state, label=above:{}, accepting] (2) at (10,3) {$q_2$};
 \node[state, label=above:{}] (4) at (5,-3) {$q_4$};
 \node[state, label=above:{}] (3) at (10,-3) {$q_3$};
 \node[state, label=above:{}, accepting] (5) at (15,-3) {$q_5$};

\draw (0) edge [above] node [above] {$a$} (1);
\draw (1) edge [below] node [below] {$c$} (2);
\draw (2) edge [loop above] node [above] {$c$} (N1);
\draw (0) edge [bend left=0, above] node [bend right, above] {$d$} (4);
\draw (4) edge [above] node [above] {$c$} (3);
\draw (3) edge [loop below] node [above, xshift=8] {$c$} (3);
\draw (3) edge[above] node [above] {$f$} (5);
\draw (4) edge[bend left=40, above] node {$f$} (5);
\end{tikzpicture}
\end{center}
    \caption{A WDFA's $\mathcal A$ recognizing the language $\mathcal L_d = ac^*+dc^*f$. Condition (i) of Definition \ref{WheelerAutomaton} implies input consistency and induces the partial order $q_1 < q_2, q_3 < q_4 < q_5$. From condition (ii) it follows that $\delta(q_1, c) \le \delta(q_4, c)$, thus $q_2 < q_3$. Therefore, the only order that could make $\mathcal A$ Wheeler is $q_0 < q_1 < q_2 < q_3 < q_4 < q_5$. The reader can verify that condition (ii) holds for each pair of equally labeled edges.}%
    \label{w}%
\end{figure}


\begin{remark}
Note that, for a fixed (i.e. constant in size) alphabet, requiring an automaton to be \emph{input-consistent} is not computationally demanding. 
In fact, given an NFA $\mathcal A=(Q,q_0,\delta,F,\Sigma)$ we can build an equivalent, input-consistent one just by creating, for each state $q\in Q$, at most $|\Sigma|$ copies of $q$, that is, one for each different incoming label of $q$.
This operation can be performed in $O\big(|Q|\cdot |\Sigma|\big)$ time.
\end{remark}

In \cite{Gagie} it is shown that WDFA's have a property called \emph{path coherence}: let $\mathcal A = (Q,q_0,\delta,F,\Sigma)$ be a WDFA according to the order $(Q,<)$. 
Then for every interval of states $I=[q_i, q_j]$ and for all $\alpha \in \Sigma^*$, the set $J$ of states reachable starting from any state of $I$ by reading $\alpha$ is also an interval.
\emph{Path coherence} allows us to transfer the order < over the states of $Q$ to the co-lexicographic order $\prec$ over the strings entering the states: two states $q$ and $p$  satisfy $q < p$  if and only if  $ \forall \alpha \in I_q \; \forall \beta \in I_p (\alpha\prec \beta)$  holds (again proved in \cite{ADPP}).  


A  consequence of this fact is that a WDFA admits an unique order of its states that makes it Wheeler and this order is univocally determined by the co-lexicographic order of any string entering its states (the order $\leq_\mt D$ mentioned in the introduction).
This result is important for two different reasons. First of all, it makes possible to decide in polynomial time whether a DFA is Wheeler: for each state $q$, pick a string $\alpha_q$ entering it and order the states reflecting the co-lexicographic order of the strings \{$\alpha_q:\, q \in Q\}$; then check if the order satisfies the Wheeler conditions. 
Secondly, it is the key to adapt Myhill-Nerode Theorem to Wheeler automata. We recall the following definition.

\begin{definition}[Myhill-Nerode equivalence]
\label{equivL}
Let $\mathcal L \subseteq \Sigma^*$ be a language. Given a string $\alpha \in \Sigma^*$, we define the \emph{right context} of $\alpha$ as  
\[
\alpha^{-1}\mathcal L := \{\gamma \in \Sigma^*:\, \alpha\gamma \in \mathcal L\},
\]
and we denote by $\equiv_\mathcal L$ the Myhill-Nerode equivalence on $\pf L$ defined as
\[
\alpha \equiv_\mathcal L \beta \iff \alpha^{-1}\mathcal L = \beta^{-1}\mathcal L.
\]
\end{definition}

The (classic) Myhill-Nerode Theorem, among many other things, establishes a bijection between equivalence classes of $\equiv_\mathcal L$ and the states of the minimum DFA recognizing $\mathcal L$. This minimum automaton is also unique up to isomorphism and a similar result, fully proved in \cite{ADPP2}, holds for Wheeler languages as well. 
In order to state such an analogous of Myhill-Nerode Theorem for  Wheeler languages, the equivalence $\equiv_\mathcal L$ is replaced by the equivalence $\equiv_\mathcal L^c$ defined below.

\begin{definition}
The input consistent, convex refinement $\equiv_\mathcal L^c$ of $\equiv_\mathcal L$ is defined as follows. $\alpha \equiv_\mathcal L^c \beta$ if and only if
\begin{itemize}
    \item $\alpha \equiv_\mathcal L \beta$,
    \item $\alpha$ and $\beta$ end with the same character,
    \item for all $\gamma \in \pf L$, if $\min(\alpha, \beta) \preceq \gamma \preceq \max(\alpha,\beta)$, then $\alpha \equiv_\mathcal L\gamma \equiv_\mathcal L \beta$.
\end{itemize}
\end{definition}

The  Myhill-Nerode Theorem for Wheeler languages proves that there exists a minimum (in the number of states) WDFA recognizing $\mathcal L$. As in the classic case,  states of the minimum automaton are, in fact, $\equiv_\mathcal L^c$-equivalence classes, this time consisting of \emph{intervals} of strings. Also, such WDFA is unique up to isomorphism.   

 \begin{theorem}\label{wdeterminization}  (see \cite{ADPP2}) 
If $ \mathcal A=(Q, s, \delta, <,F) $ is a WNFA with $|Q|=n$  and $\mathcal L = \mathcal L(\mathcal A)$,   then there exists a unique minimum-size  WDFA $ \mathcal B $ with  $2n-1-|\Sigma|$ states such that $\mathcal L= \mathcal L(\mathcal B)$.
\end{theorem}

Starting from the (possibly non Wheeler) minimum DFA of  a Wheeler language $\mathcal L$, we will give an algorithm   constructing   the minimum Wheeler  automaton for the language. This automaton   can be described   as follows (see \cite{ADPP2}): $\mathcal B =(Q', \delta', q_0', F')$ where \\ \ \\ - $Q' = \{[\alpha]_{\equiv_\mathcal L^c}  : \alpha\in \pf L\}$;\\ 
- $q_0 =  [\epsilon]_{\equiv_\mathcal L^c}$;\\
- $\delta'([\alpha]_{\equiv_\mathcal L^c}, a)=[\alpha a]_{\equiv_\mathcal L^c}$,~ for all $\alpha \in \pf L$, $a\in \Sigma$,\\
- $F'=\{[\alpha]_{\equiv_\mathcal L^c}
  :  \alpha \in \mathcal L\}$.

\section{Reduced NFA's meets Wheelerness}

\subsection{Automata}
Among the two possible ways of presenting regular languages by automata, that is DFA's or  NFA's, in general, computational problems tend to be significantly harder when referred to the non-deterministic class. Typical examples are: checking emptiness, computing the intersection, checking universality and much more. 
In the realm of Wheeler automata and languages a new class emerges: the class of reduced automata, formally defined below.


\begin{definition}
An NFA $\mathcal A=(Q,S,\delta, F,\Sigma)$ is called \emph{reduced} if $q\ne p$ implies $I_q \ne I_p$.
\end{definition}
Clearly, the class of reduced NFA's contains  properly the class of DFA's.
When Wheelerness is concerned, the class of reduced NFA's is interesting  because it has been proved that deciding whether an NFA is Wheeler is an NP-complete problem \cite{NP}, whereas deciding whether a \emph{reduced} NFA is Wheeler turns out to be in P \cite{ADPP2} as it is for DFA's \cite{ADPP}. 
Clearly, any NFA can be turned into a reduced one simply by merging all the states that recognize the same incoming language.
Finding states to be merged is complex: the language-equivalence problem for NFA's  can easily be proved  as complex as deciding whether two states of an NFA recognize the same incoming language and, therefore, the latter is PSPACE-complete. 
 


A natural question is now whether switching from NFA's to  reduced NFA's simplifies some otherwise difficult problem. 
In this section we prove that this is not always the case: some problems remain hard even when restricted to the class of reduced NFA's.

\begin{lemma}
\label{reduced universality}
The universality problem for reduced NFA's is PSPACE-complete.
\end{lemma}
\begin{proof}
This problem belongs to PSPACE, since it is a restriction of the universality problem over generic NFA's. 
To prove the completeness, we show a reduction from the universality problem.

Given an NFA, we can assume w.l.o.g. that there is only one initial state without incoming edges, hence 
let $\mathcal A=(Q,q_0,\delta, F,\Sigma)$ be an NFA with $Q=\{q_0,\dots,q_n\}$ be such an NFA. 
We build a new automaton $\mathcal A'=(Q\cup P,q_0,\delta', F,\Sigma\cup\{d\})$, where $P=\{p_1,\dots,p_{n-1}\}$ is a set of $n-1$ new states and $d$ is a new character. 
For each $q\in Q$ we add the self loop $(q,d,q)$. 
If we add only these transitions, it holds that $\la A=\Sigma^*$ iff $\la {A'}=(\Sigma + d)^*$. 
We can now add to the automaton as many $d$-transitions as we please without violating the property $\la A=\Sigma^*$ iff $\la {A'}=(\Sigma + d)^*$: the right-to-left implication still holds if we only add $d$-transitions, whereas the left-to-right implication holds since adding transitions may only expand the recognized language, but $(\Sigma + d)^*$ is already maximal (with respect to the inclusion). 
Therefore we add the transitions $(q_0, d,q_1)$ and $(q_0,d,p_1)$.
Moreover, for each $1\le i \le n-1$ we add the transitions $(p_i, d,q_{i+1})$ and $(p_i, d,p_{i+1})$ (see Figure \ref{trans d}).

\begin{figure}
\begin{center}
\begin{tikzpicture}[->,>=stealth', semithick, initial text={}, auto, scale=.35]
\node[state, label=above:{}, minimum size=1.15cm] (0) at (0,0) {$q_0$};
\node[state, label=above:{}, minimum size=1.15cm] (1) at (6,0) {$q_1$};
\node[state, label=above:{}, minimum size=1.15cm] (2) at (12,0) {$q_2$};
\node[state, label=above:{}, minimum size=1.15cm] (3) at (18,0) {$q_3$};
\node[state, label=above:{}, draw=none] (d) at (24,0) {$\dots$};
\node[state, label=above:{}, minimum size=1.15cm] (n) at (30,0) {$q_n$};
\node[state, label=above:{}, minimum size=1.15cm] (p1) at (6,-6) {$p_1$};
\node[state, label=above:{}, minimum size=1.15cm] (p2) at (12,-6) {$p_2$};
\node[state, label=above:{}, draw=none] (pd) at (18,-6) {$\cdots$};
\node[state, label=above:{}, minimum size=1.15cm] (pn) at (24,-6) {$p_{n-1}$};

\draw (0) edge [loop above] node {$d$} (0);
\draw (1) edge [loop above] node {$d$} (0);
\draw (2) edge [loop above] node {$d$} (0);
\draw (3) edge [loop above] node {$d$} (0);
\draw (n) edge [loop above] node {$d$} (0);

\draw (0) edge node {$d$} (1);
\draw (0) edge node {$d$} (p1);
\draw (p1) edge node {$d$} (p2);
\draw (p2) edge node {$d$} (pd);
\draw (pd) edge node {$d$} (pn);
\draw (p1) edge node {$d$} (2);
\draw (p2) edge node {$d$} (3);
\draw (pn) edge node {$d$} (n);

\end{tikzpicture}
\end{center}
\caption{The automaton $\mathcal A'$ with only $d$-transitions depicted.}
\label{trans d}
\end{figure}
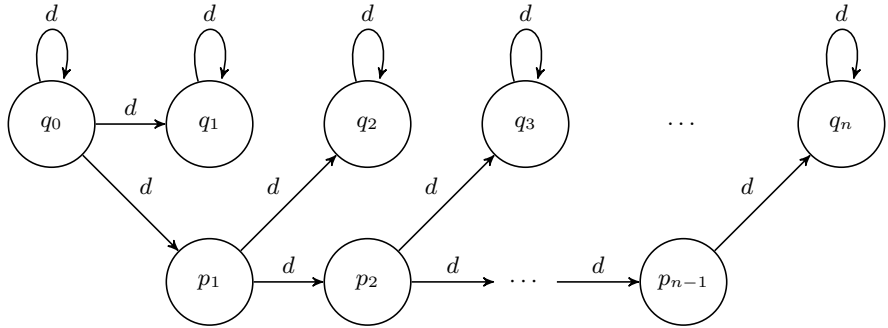
To conclude the proof that the reduction is correct, we need to show that $\mathcal A'$ is reduced. Since $q_0$ has no incoming edges, we have \begin{align*}
    &I_{q_0}=d^*\\
    &I_{p_i}=d^i\cdot d^*\text{\quad for }1 \le i\le n-1.
\end{align*}
Since $\mathcal A$ was trimmed and since each $q\in Q\setminus\{q_0\}$ is not an initial state, we have $I_q\cap \Sigma^+ \ne \emptyset$ for each $q\in Q\setminus\{q_0\}$. Thus $I_q\ne I_p$ for each $q\in Q\setminus\{q_0\}$ and for each $p\in P\cup\{q_0\}$. Moreover, for each $1\le i< j \le n$ we have $d^i \in I_{q_i}\setminus I_{q_j}$, hence $I_{q_i}\ne I_{q_j}$.
\end{proof}

\

We will use the previous lemma to solve a problem related to another interesting aspect of the relationship between DFA's, NFA's, and reduced NFA's: indexability. 
Given an NFA $\mathcal A$, it is possible to define a partial order $<_{\mathcal A}$ on its states that allows to represent $\mathcal A$ using an index, that is, a succint structure that supports fast matching queries \cite{JACM}. 
The partial order $<_{\mathcal A}$ is defined using the family of incoming languages $\{I_q: q\in Q\}$. As opposed to the case of DFA's, over NFA's these languages may not be pairwise disjoint,  and we can compare them as follows:

\[
I_q\preceq I_p \iff \forall \alpha\in I_q\; \forall\beta\in I_p\big(\{\alpha, \beta\} \not\subseteq I_q\cap I_p \Rightarrow \alpha \prec \beta\big).
\]
The above  partial order can be lifted to the collection of states of an NFA.  

\begin{definition}
\label{prec nfa}
Given two states $q$ and $p$ of an NFA $\mathcal A$, we say that $q<_{\mathcal A} p$ iff $I_q\preceq I_p$ and $I_q\neq I_p$.
\end{definition}

Note that if $\mathcal D$ is a DFA, then $<_{\mathcal D}$  simplifies:
\[
q<_{\mathcal D} p \iff \forall \alpha\in I_q\;\forall\beta\in I_p\big(\alpha \prec \beta\big),
\]
and this order satisfies the properties of a Wheeler order, with the exception of not necessarily being total. As a matter of fact, it can be proved that the DFA $\mathcal D$ is Wheeler if and only if $<_{\mathcal D}$ is a total order. Remarkably, this partial order can be computed in polynomial time \cite{JACM} on DFA's.

\begin{proposition}
Let $\mathcal D$ be a DFA with $n$ states. Then, we can compute the order $<_{\mathcal D}$ in $O(n^5)$ time.
\end{proposition}

It follows that,  given a DFA $\mt D$,  we can compute $<_\mathcal D$ in polynomial time and use it to index $\mathcal D$ efficiently.
Would it be possible to generalized this result  to NFA's using the corresponding partial order  $<_\mathcal A$ of Definition \ref{prec nfa}?
In the following lemma we give a negative answer to this question, 
even when restricted to reduced automata, proving that a different approach is needed to index NFA's (see \cite{JACM} for a positive solution to the problem).

\begin{theorem}
\label{prec not eq}
Given two states $q$ and $p$ of an NFA $\mathcal A$, deciding whether $q <_{\mathcal A} p$ is PSPACE-complete. 
The same result holds even if $\mathcal A$ is reduced.
\end{theorem}
\begin{proof}
First of all we need to prove that the problem is in PSPACE. We will show instead that its complement is in PSPACE and the thesis follows from the fact that PSPACE is closed under complementation. 
The complement of our problem consist of answering to the question whether $q\nless p$.
To do so, first we check whether $I_q=I_p$. As we have already mentioned, this problem is in PSPACE, so we can get the answer in polynomial space. 
If $I_q=I_p$, then $q\nless p$ and we answer "yes". Otherwise, we have
\[
q<_{\mathcal A} p \iff \forall \alpha\in I_q\; \forall\beta\in I_p\big( \{\alpha, \beta\} \not\subseteq I_q\cap I_p \Rightarrow \alpha \prec \beta \big),
\]
or equivalently 
\[
q\nless_{\mathcal A} p \iff \exists \alpha\in I_q\; \exists\beta\in I_p\big( \{\alpha, \beta\} \not\subseteq I_q\cap I_p \wedge \beta \prec \alpha \big).
\]
Let $d$ be the number of states of the DFA $\mathcal D$ generated by the determinization of $\mathcal A$; clearly it holds $d\le 2^n$. 
We claim that if $q\nless p$, then there exist two strings $\alpha, \beta$ of length at most $d^2+d$ such that
\begin{equation}
\label{eq}
    \alpha\in I_q\;\wedge\; \beta\in I_p\;\wedge\; \{\alpha, \beta\} \not\subseteq I_q\cap I_p \;\wedge \;\beta \prec \alpha.
\end{equation}
Assume that $\alpha, \beta$ satisfy \eqref{eq}, with either $|\alpha|$ or $|\beta|$ (possibly both) greater than $d^2+d$.
We assume, w.l.o.g., that $|\alpha|\le |\beta|$ and distinguish two cases.
\\1) The last $d^2$ characters of $\alpha$ and $\beta$ differs; this also includes the case where $|\alpha|$ is strictly less than $d^2$. 
Consider the $d+1$ states of $\mathcal D$ visited by reading the first $d$ characters of $\beta$. Since $\mathcal D$ has $d$ states, at least one of them appears twice, implying that we visited a cycle.
By erasing from the first $d$ characters of $\beta$ the factor corresponding to such cycle, we obtain a string $\beta'$ such that $\alpha$ and $\beta'$ also satisfy \eqref{eq}.
\\2) The last $d^2$ characters of $\alpha$ and $\beta$ coincide; in particular $|\alpha|,|\beta|\ge d^2$. Consider the last $d^2+1$ states $r_0, ..., r_{d^2}$ of $\mathcal D$ visited by reading the string $\alpha$, and the last $d^2+1$ states $p_0, ..., p_{d^2}$ visited by reading the string $\beta$. 
Since $\mathcal D$ has only $d$ states, there must exist $0 \le i,j \le d^2$ with $i < j$ such that $(r_i, p_i) = (r_j, p_j)$, implying that $\alpha$ and $\beta$ visited two cycles labeled by the same string.
By erasing from the last $d^2$ characters of $\alpha$ and $\beta$ the factor corresponding to such cycles, we obtain two strings $\alpha', \beta'$ which also satisfy \eqref{eq}.
\\In both cases, we were able to shorten the length of the longest string. By repeating this process as many times as needed, we will eventually obtain two strings both shorter than $d^2+d$, with $d\le 2^n$.

Now that we have bounded the length of $\alpha, \beta$ with the constant $2^{2n}+2^n$, we can use non-determinism to guess, bit by bit, the length of $\alpha$ and $\beta$ and store this guessed information in two counters $a, b$ respectively, using $O\big(\log (2^{2n}+2^n)\big)= O(n)$ space for each. 
These counters determine which string among $\alpha$ and $\beta$ is longer and we start guessing the characters of such longest string from the left to the right, decreasing by one its counter whenever we guess a character. 
Note that we are not storing the guessed characters, since it would use too much space. 
When the counter reaches the same value of the other counter, we start guessing the characters of both the first and the second string at the same time and we carry on until both counters reach the value 0.
While guessing the characters of $\alpha$ (respectively, $\beta$) we update at each step the set of states of $\mathcal A$ reachable from $q_0$ by reading the currently guessed prefix of $\alpha$ ($\beta$), so that in the end we obtain the sets $\delta(q_0, \alpha)$ and $\delta(q_0, \beta)$.
With this information, we can check whether $\alpha\in I_q$ and $\beta\in I_p$ and $\{\alpha, \beta\} \not\subseteq I_q\cap I_p$.
To complete checking condition \eqref{eq}, we need to show how to decide whether $\beta\prec\alpha$.  

To confront co-lexicographically $\alpha$ and $\beta$, we use a variable $\rho$ that indicates whether $\alpha$ is less, equal or greater than $\beta$.
We initialize $\rho$ based on the counters $a,b$ as follows:
\[
\rho:=\begin{cases}
= \quad &\text{if }a=b\\
\dashv \quad &\text{if }a<b\\
\vdash \quad &\text{if }b<a.
\end{cases}
\]
We leave $\rho$ unchanged until we start guessing simultaneously the characters of $\alpha$ and $\beta$. 
When we guess the character $c_1$ for $\alpha$ and the character $c_2$ for $\beta$, we set
\[
\rho:=\begin{cases}
\prec \quad &\text{if }c_1\prec c_2\\
\succ \quad &\text{if }c_1\succ c_2\\
\rho \quad &\text{if }c_1=c_2.
\end{cases}
\]
Note that if at the end $\rho$ has value $\dashv$, it means that $\alpha \dashv \beta$, thus $\alpha \prec\beta$.
Similarly, if $\rho$ has value $\vdash$ then $\beta \prec\alpha$. Otherwise, we have $\alpha \,\rho\, \beta$. Thus we are always able to determine the co-lexicographic order of $\alpha$ and $\beta$. 
Therefore, deciding whether $q\nless p$ is a problem in PSPACE, and so it is its complement.

To prove  completeness, we show a reduction from the universality problem over generic, respectively reduced, NFA's. 

 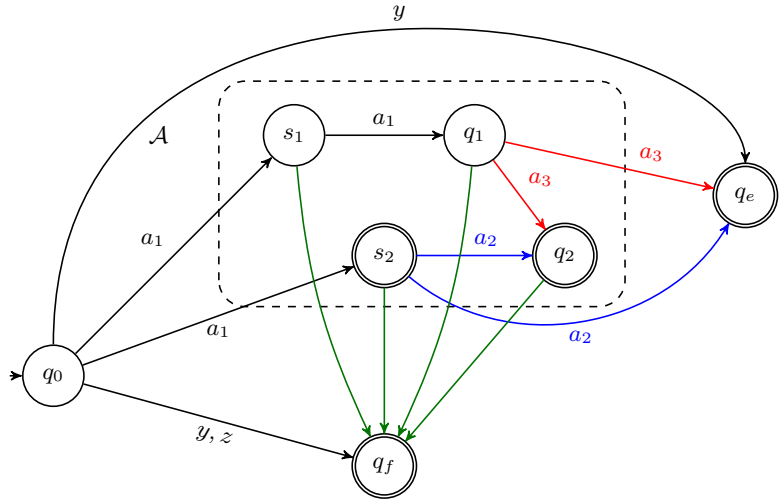
\begin{figure}[H]
\begin{center}
\begin{tikzpicture}[->,>=stealth', semithick, initial text={}, auto, scale=.4]
\node[state, label=above:{},initial] (0) at (-8,0) {$q_0$};

\node[state, label=above:{}] (1) at (0,8) {$s_1$};
\node[state, label=above:{}, accepting] (2) at (3,4) {$s_2$};
\node[state, label=above:{}] (3) at (6,8) {$q_1$};
\node[state, label=above:{}, accepting] (4) at (9,4) {$q_2$};

\node[state, label=above:{}, accepting] (5) at (15,6) {$q_e$};
\node[state, label=above:{}, accepting] (6) at (3,-3) {$q_f$};
\node (7) at (-4.5,8) {$\mathcal A$};

\draw[dashed, rounded corners = 10] (-2.5,2.3) rectangle (11,9.8);

\draw (0) edge node {$a_1$} (1);
\draw (0) edge[below] node {$a_1$} (2);
\draw (1) edge node {$a_1$} (3);
\draw (2) edge[color=blue, pos=.6] node {$a_2$} (4);
\draw (2) edge[bend right=50, below, color=blue] node {$a_2$} (5);
\draw (3) edge[color=red] node {$a_3$} (4);
\draw (3) edge[color=red] node[pos=.6] {$a_3$} (5);

\draw (1) [cgreen] edge[bend right=10] node {} (6);
\draw (2) [cgreen] edge node {} (6);
\draw (3) [cgreen] edge[bend left=10] node {} (6);
\draw (4) [cgreen] edge node {} (6);
\draw (0) edge[below, sloped] node {$y, z$} (6);
\draw (0) .. controls +(0,16) and +(0,6) .. node {$y$} (5);

\end{tikzpicture}
\end{center}
\caption{The automaton $\mathcal A'$ built starting from the automaton $\mathcal A$ with $S=\{s_1, s_2\}$ recognizing the language $\mathcal L=\{\varepsilon, a_2, a_1a_3\}$. Edges entering a final state in $\mathcal A$ have been duplicated and redirected to $q_e$. Green edges are labeled $\Sigma=\{a_1, a_2, a_3\}$.} 
\label{fig:reduced}
\end{figure}

Let $\mathcal A=(Q,S,\delta, F,\Sigma)$ be an NFA with $Q=\{q_1,\dots,q_n\}$ and $\Sigma=\{a_1,\dots,a_\sigma\}$ recognizing the language $\mathcal L = \la A$, we build a new NFA $\mathcal A'=(Q',q_0,\delta', F\cup\{q_e,q_f\},\Sigma')$ by adding a new initial state $q_0$ and two final states $\{q_e,q_f\}$ (see Figure \ref{fig:reduced}). 
The new alphabet is $\Sigma'=\Sigma \cup \{y,z\}$, where $a_j \prec y \prec z$ for each $1 \le j \le \sigma$.
For each $q_i\in S$, we add a transition from $q_0$ to $q_i$ labeled $a_1$. 
Adding $q_0$ has the sole purpose of having an initial state without incoming edges.  Note that we can not make the usual assumption that $\mathcal A$ has only one initial states without incoming edges: if we start from a reduced NFA and we build an equivalent NFA with the required property, there is no guarantee that the new automaton will still be reduced. 
The state $q_e$ represents the new final state that gathers all the strings in $a_1\cdot(\mathcal L \setminus \{\varepsilon\})$. 
To achieve this goal, for each transition $(q_i, a_j, q_{i'})$ of $\delta$ such that $q_{i'} \in F$ we add a new transition $(q_i, a_j, q_e)$.
The state $q_f$ gathers all the strings in $a_1\cdot \pf{L} \cdot \Sigma$, and this can be easily achieved by adding a transition $(q_i, a_j, q_f)$ for each $i \ge 1$ and $j \ge 1$. 
Lastly, we add the transitions $(q_0,y,q_e)$, $(q_0,y,q_f)$ and $(q_0,z,q_f)$. 
This way, if $\mathcal A$ is reduced then $\mathcal A'$ is also reduced: note that $I_{q_0}=\{\varepsilon\}$, for each $i \ge 1$ it holds $I^{\mathcal A'}_{q_i} = a_1\cdot I^{\mathcal A}_{q_i}$, the states $q_e,q_f$ are the only that can read the string $y$  and $q_f$ is the only state that can read the string $z$.

Let $\mathcal L_\varepsilon$ denote the language $\mathcal L \setminus \{\varepsilon\}$. By construction, we have
\begin{align*}
    I_{q_e}&= a_1\cdot \mathcal L_\varepsilon + y \\
    I_{q_f}&= a_1\cdot \pf L \cdot \Sigma + y + z.
  \end{align*}
We want to show that $\mathcal L = \Sigma^*$ iff $q_e < q_f\, \wedge\, \Sigma \subseteq \mathcal L_\varepsilon$. Note that $\Sigma \subseteq \mathcal L_\varepsilon$ is a necessary condition for $\mathcal L$ to be universal, and such condition can be checked in polynomial time using reachability on $\mathcal A$, therefore the reduction is still polynomial.
\\$(\Rightarrow)$ If $\mathcal L = \Sigma^*$, it clearly follows that $\Sigma \subseteq \mathcal L_\varepsilon$. Moreover we have $\pf L \cdot \Sigma = \Sigma^+$ and we obtain
\begin{align*}
    I_{q_e}&= a_1\cdot \Sigma^+ + y \\
    I_{q_f}&= a_1\cdot \Sigma^+ + y + z.
\end{align*}
It follows immediately that $q_e < q_f$.
\\$(\Leftarrow)$ 
Note that $\mathcal L_\varepsilon \subseteq \pf L \cdot \Sigma$.
We first prove that from the hypothesis it follows $\mathcal L_\varepsilon = \pf L \cdot \Sigma$. Assume by contradiction that $\mathcal L_\varepsilon \ne \pf L \cdot \Sigma$ and let $\beta$ be a string in $\pf L \cdot \Sigma \setminus \mathcal L_\varepsilon$.
Then we have
\[
y\in I_{q_e}, \quad a_1\cdot\beta \in I_{q_f}, \quad \{y, a_1\cdot\beta\} \nsubseteq I_{q_e} \cap I_{q_f}
\]
but $y \succ a_1\cdot\beta$, a contradiction. Thus $\mathcal L_\varepsilon = \pf L \cdot \Sigma$. 

We can then prove by induction on $|\alpha|$ that $\alpha \in \Sigma^+$ implies $\alpha\in\mathcal L_\varepsilon$. 
If $|\alpha|=1$ then $\alpha\in\Sigma$ and by hypothesis we have $\Sigma \subseteq \mathcal L_\varepsilon$. 
If $|\alpha|=n+1>1$, then $\alpha = \alpha'\cdot a_j$ for some $\alpha' \in \Sigma^+$ and some $a_j \in \Sigma$. 
By induction hypothesis we have $\alpha' \in \mathcal L_\varepsilon \subseteq \pf L$, and from $\mathcal L_\varepsilon = \pf L \cdot \Sigma$ it follows $\alpha \in \mathcal L_\varepsilon$.

This concludes the reduction from the universality problem to our problem over general NFA's. Since the construction described preserves the reduced-ness of the starting automaton, it also works as a reduction from the universality problem over reduced NFA's to our problem over reduced NFA's. In Lemma \ref{reduced universality} we proved that the former problem is PSPACE-complete, thus proving that the latter is also PSPACE- complete.
\end{proof}

\

We can use the previous results to prove another complexity result over reduced NFA's.

\begin{corollary}
Deciding whether an NFA $\mathcal A$ is reduced is PSPACE-complete.
\end{corollary}
\begin{proof}
To prove that the problem is in PSPACE, note that $\mt A= (Q, q_0, \delta, F, \Sigma)$ is reduced iff, for all $q,p\in Q$, $q\ne p$ implies $I_q\ne I_p$. Therefore, it is sufficient to check $O(n^2)$ times whether $I_q=I_p$, where $n=|Q|$.
As we have already mentioned, the problem of deciding whether $I_q=I_p$ belongs to PSPACE, thus the thesis follows.

To prove  completeness, we combine the reductions shown in Lemma \ref{reduced universality} and Theorem \ref{prec not eq}.
Let $\Sigma_d=\Sigma \cup \{d\}$. 
We first apply the reduction shown in Lemma \ref{reduced universality} to build a \emph{reduced} automaton $\mathcal A'$ such that $\la A = \Sigma^*$ iff $\la{A'}=\Sigma_d^*$. We set $\mathcal L':= \la{A'}$.
Then, we apply the reduction showed in Theorem \ref{prec not eq} to the automaton $\mathcal A'$, but we remove the edge $(q_0, z, q_f)$; we call this new automaton $\mathcal A''$. 
The languages recognized by $q_e$ and $q_f$ change as follow:
\begin{align*}
    I_{q_e}&= a_1\cdot \mathcal L'_\varepsilon + y \\
    I_{q_f}&= a_1\cdot \pf {L'} \cdot \Sigma + y.
\end{align*}
Since $\mathcal A'$ is a reduced automaton and the states $q_e$ and $q_f$ are the only ones with an incoming edge labeled $y$, it immediately follows that $\mathcal A''$ is \emph{not} reduced iff $I_{q_e}=I_{q_f}$.
Applying the same argument we used in Theorem \ref{prec not eq}, we can conclude that $\la A'=\Sigma_d^*$ iff $I_{q_e}=I_{q_f}$ ---again, we assumed that $\Sigma_d \subseteq \la{A'}$, since this condition can be checked in polynomial time.
Summarizing we have that $\mathcal A''$ is \emph{not} reduced iff $\la A = \Sigma^*$. 
Our claim follows from the equality PSPACE=NPSPACE.
\end{proof}

Note that, as proved in \cite{ADPP2}, deciding whether  a \emph{Wheeler} NFA is reduced is a simpler problem, being  in P. 

\subsection{Languages}

In this section we switch our focus from automata to languages.
An important consequence of the Myhill-Nerode Theorem for Wheeler languages is stated in the following Lemma (proved in \cite{ADPP2}).

\begin{lemma}
\label{monotone}
A regular language $\mathcal L$ is Wheeler if and only if all monotone sequences in $(\pf L, \prec)$ become eventually constant modulo $\equiv_\mathcal L$. In other words, for all sequences $(\alpha_i)_{i \ge 0}$ in $\pf L$ with
\[
\alpha_1 \preceq \alpha_2 \preceq \dots \alpha_i \preceq \dots \quad \text{ or }\quad \alpha_1 \succeq \alpha_2 \succeq \dots \succeq \alpha_i \succeq \dots
\]
there exists an $n$ such that $\alpha_h \equiv_\mathcal L \alpha_k$, for all $h,k \ge n$. 
\end{lemma}

Lemma \ref{monotone} shows how it is possible to recognize whether a language $\mathcal L$ is Wheeler simply by verifying a property on elements of $\pf L$: trying to find a WDFA that recognizes $\mathcal L$ is no longer needed to decide  Wheelerness of $\mathcal L$. 
As shown in Theorem \ref{polynomialW} (see \cite{ADPP2}), we can verify whether the property mentioned in Lemma \ref{monotone} is satisfied just analysing the structure of the minimum DFA recognizing $\mathcal L$.

\begin{theorem}
\label{polynomialW}
Let $\mathcal D_\mt L$ be the minimum DFA that recognizes the language $\mathcal L$,  with initial state $q_0$ and dimension $n = |\mathcal D_\mt L|$. 
\\$\mathcal L$ is not Wheeler if and only if there exist $\mu, \nu$ and $\gamma$ in $\Sigma^*$, with $\gamma \ndashv \mu,\nu$, such that:
\begin{enumerate}
    \item $\mu \not\equiv_\mathcal L \nu$ and they label paths from $q_0$ to states $u$ and $v$, respectively;
    \item $\gamma$ labels two cycles, one starting from $u$ and one starting from $v$;
    \item $\mu, \nu \prec \gamma$\; or \; $\gamma \prec \mu,\nu$.
\end{enumerate}
The length of the strings $\mu, \nu$ and $\gamma$ satisfying the above  can be bounded: 
\begin{enumerate}
\setcounter{enumi}{3}
    \item $|\mu|, |\nu| \le |\gamma| \le n^3+2n^2+n+2$.
\end{enumerate}
\end{theorem}

The proof of Theorem \ref{polynomialW} in \cite{ADPP2} can be adapted to work on generic DFA's. 
Since such proof is both long and technical, we will prove instead (in the Appendix) the following proposition, where we worsen the bound given in condition 4. 
This is not a problem, since we will only use the fact that this bound is polynomial in $n$.

\begin{proposition}
\label{general polynomialW}
Let $\mathcal D = (Q, q_0, \delta, F, \Sigma)$ be a DFA recognizing the language $\mathcal L$,  with $n = |\mathcal D|$. 
\\$\mathcal L$ is not Wheeler if and only if there exist $\mu, \nu$ and $\gamma$ in $\Sigma^*$, with $\gamma \ndashv \mu,\nu$, such that:
\begin{enumerate}
    \item $\mu \not\equiv_\mathcal L \nu$ and they label paths from $q_0$ to states $u$ and $v$, respectively;
    \item $\gamma$ labels two cycles, one starting from $u$ and one starting from $v$;
    \item $\mu, \nu \prec \gamma$\; or \; $\gamma \prec \mu,\nu$.
\end{enumerate}
The length of the strings $\mu, \nu$ and $\gamma$ satisfying the above  can be bounded:
\begin{enumerate}
\setcounter{enumi}{3}
    \item $|\mu|, |\nu| \le |\gamma| \le (n^3+2n^2+n+2)\cdot n^2$.
\end{enumerate}
\end{proposition}

The polynomial bound given by condition 4 of Theorem \ref{polynomialW} allows us to design an algorithm that decides whether a given DFA recognizes a Wheeler language: using dynamic programming (see \cite{ADPP}) it is possible to keep track of all the relevant paths and cycles inside the DFA and check, in polynomial time, whether there exists three strings satisfying the conditions of the theorem.  

Things change if, instead of a DFA, we are given an NFA. Trying to exploit the same idea used for DFA's does not work: the problem of deciding whether two strings $\mu$ and $\nu$ read by an NFA are Myhill-Nerode equivalent is PSPACE-complete. Even worse, a straightforward attempt of building the minimum DFA recognizing the NFA's language might lead to a blow-up of the sates, resulting in a exponential time (and exponential space) algorithm. 

We show that the problem of deciding whether an NFA recognizes a Wheeler language is indeed hard, but does not necessarily require exponential time to be solved. Instead, the problem turns out to be PSPACE-complete. To show this, we first show how to adapt Theorem \ref{polynomialW} to work on NFA's, as described in the following corollary.

\begin{corollary}
\label{nfa length}
Let $\mathcal A = (Q, q_0, \delta, F, \Sigma)$ be an NFA of dimension $n := |\mathcal A|$. Then $\mathcal L := \la A$ is not Wheeler if and only if there exist three strings $\mu, \nu, \gamma$ such that $\gamma\ndashv\mu,\nu$ and
\begin{enumerate}
    \item $\mu\gamma^i \not\equiv_{\mathcal L} \nu\gamma^j$ for all $0 \le i,j \le 2^n$; 
    \item $\gamma$ labels two cycles, one starting from a state $p \in \delta(q_0,\mu)$ and one from a state $r \in \delta(q_0,\nu)$;
    \item $\mu, \nu \prec \gamma$ or $\gamma \prec \mu, \nu$.
\end{enumerate}
Moreover, the  length of the strings $\mu, \nu$ and $\gamma$ satisfying the above can be bounded:
\begin{enumerate}
\setcounter{enumi}{3}
    \item $|\mu|, |\nu| < |\gamma| < n^3\cdot(2^{3n}+2\cdot 2^{2n}+2^n+2)\in O(2^{3n})$.
\end{enumerate}
\end{corollary}
\begin{proof}
Let $\mathcal D = (\hat Q, \hat q_0, \hat \delta, \hat F, \Sigma)$ be the minimum DFA recognizing $\mathcal L$. Clearly $\mathcal D$ has at most $2^n$ states.
\\($\Longleftarrow$) From condition 2 it follows that $\mu\gamma^* \subseteq \pf L$, so consider the following list of $2^n+1$ states of $\mathcal D$:
\[
\hat\delta(\hat q_0, \mu\gamma^0), \, \hat\delta(\hat q_0, \mu\gamma^1), \, \dots, \, \hat\delta(\hat q_0, \mu\gamma^{2^n}).
\]
Since $\mathcal D$ has at most $2^n$ states, there must exist two integers $0 \le h < k \le 2^n$ such that $\hat\delta(\hat q_0, \mu\gamma^h) = \hat\delta(\hat q_0, \mu\gamma^k)$. Therefore $\gamma^{k-h}$ labels a cycle starting from $\hat\delta(\hat q_0, \mu\gamma^h)$. Similarly, there exist $0 \le h' < k' \le 2^n$ such that $\gamma^{k'-h'}$ labels a cycle starting from $\hat\delta(\hat q_0, \nu\gamma^{h'})$. The strings 
\begin{align*}
\hat\mu &:= \mu \gamma^{h} \\
\hat\nu &:= \nu \gamma^{h'} \\
\hat\gamma &:= \gamma^{\text{lcm}(k-h, k'-h')\cdot 2^n},
\end{align*}
where the factor $2^n$ in the definition of $\hat\gamma$ ensures that $|\hat\mu|, |\hat\nu| < |\hat\gamma|$, so that $\hat\gamma\not \dashv \hat\mu, \hat\nu$ and the strings $\hat\mu, \hat\nu, \hat\gamma$ satisfy condition 2 of Theorem \ref{polynomialW}. 
Condition 1 of Theorem \ref{polynomialW} follows automatically from conditions 1 of this corollary. Lastly, condition 3 of Theorem \ref{polynomialW} follows from conditions 3 of this corollary and the fact that $\gamma\ndashv\mu,\nu$.
Thus we can apply Theorem \ref{polynomialW} to conclude that $\mathcal L$ is not Wheeler.
\\$(\Longrightarrow)$
Since $\mathcal L = \la D$ is not Wheeler, let $\hat\mu, \hat\nu, \hat\gamma$ be  strings satisfying  Theorem \ref{polynomialW}. The DFA $\mathcal D$ has at most $2^n$ states, hence the length of $\hat\gamma$ is bounded by the constant $2^{3n}+2\cdot 2^{2n}+2^n+2$.
We have $\hat\mu\hat\gamma^* \subseteq \pf L$, so let $t_0 = q_0, t_1, \dots, t_m$ be a run of $\hat\mu\hat\gamma^n$ over $\mathcal A$. We set $u := |\hat\mu|$ and $g := |\hat\gamma|$, and consider the list of $n+1$ states 
\[
t_u, \; t_{u+g}, \; t_{u+2g}, \; \dots, \; t_{u+ng} = t_m 
\]
Since $\mathcal A$ has $n$ states, there must exist two integers $0 \le h < k \le n$ such that $t_{u+hg} = t_{u+kg}$. That is, there exists a state $p := t_{u+hg}$ such that $p \in \delta\left(q_0, \hat\mu\hat\gamma^h\right)$ and $\hat\gamma^{k-h}$ labels a cycle starting from $p$. We can repeat the same argument for a run of $\hat\nu\hat\gamma^n$ over $\mathcal A$ to find a state $r$ and two integers $h', k'$ such that $r \in \delta(q_0, \hat\nu\hat\gamma^{h'})$ and $\hat\gamma^{k'-h'}$ labels a cycle starting from $r$. We can then define the strings
\begin{align*}
\mu &:= \hat\mu \hat\gamma^{h} \\
\nu &:= \hat\nu \hat\gamma^{h'} \\
\gamma &:= \hat\gamma^{\text{lcm}(k-h, k'-h')\cdot n}
\end{align*}
which satisfy the conditions 2 and 3.
\\Condition 4 is satisfied since $|\hat\gamma| \le 2^{3n}+2\cdot2^{2n}+2^n+2$ and $\text{lcm}(k-h,k'-h') < n^2$.
\\Finally, condition 1 is satisfied for all $i,j\ge0$. Indeed, for all $l$ the strings $\hat\mu$ and $\hat\mu\hat\gamma^l$ lead to the same state of $\mathcal D$, thus $\hat\mu \equiv_\mathcal L \hat\mu\hat\gamma^l$. Similarly, for all $l$ we also have $\hat\nu \equiv_\mathcal L \hat\nu\hat\gamma^l$. Since $\forall i\; \exists s_i$ such that $\mu\gamma^i = \hat\mu \hat\gamma^{s_i}$, and similarly, $\forall j\; \exists s_j$ such that $\nu \gamma^j= \hat\nu \hat\gamma^{s_j}$, the thesis follows from $\hat\mu \not\equiv_\mathcal L \hat\nu$. 
\end{proof}

Despite the fact the the bound in condition 4 has become exponential by switching to NFA's, it is still possible to check in polynomial space (but exponential time) whether there are three strings $\mu, \nu$ and $\gamma$ satisfying the conditions of Proposition \ref{general polynomialW}. Thus we can prove the following:

\begin{theorem}
\label{pspace}
Given an NFA $\mathcal A= (Q, q_0, \delta, F, \Sigma)$, deciding whether the language $\mathcal L := \la A$ is Wheeler is PSPACE-complete. The same result holds even if $\mathcal A$ is reduced.
\end{theorem}
\begin{proof}
First of all we need to prove that the problem is in PSPACE. We will show instead that its complement is in NPSPACE, then the thesis follows from Savitch's Theorem, which states that NPSPACE = PSPACE, and the fact that PSPACE is closed under complementation. 

Let $\mathcal D$ be the automaton obtained by the determinization of $\mathcal A$ with dimension $d=|\mathcal D|\le 2^n$. We  prove that we can check  the conditions in Proposition \ref{general polynomialW} for the automaton $\mathcal D$, without building it, using polynomial space.
We use non-determinism to guess, bit by bit, the length of $\mu, \nu$ and $\gamma$ and store this guessed information in three counters $u, v, g$ respectively, using $O\big(\log(d^5)\big)=O(n)$ space for each. 
These counters determine which string among $\mu, \nu$ and $\gamma$ is longer and we start guessing the characters of such string from the left to the right, decreasing by one its counter whenever we guess a character. 
When the counter reaches the same value of the second biggest counter, we start guessing the characters of both the first and the second string at the same time and we carry on until they reach the value of the last counter.
Then, we guess simultaneously the characters of all three strings until all counters reach the value 0.
While guessing the characters of $\mu$ (respectively, $\nu$) we update at each step the set of states of $\mathcal A$ reachable from $q_0$ by reading the currently guessed prefix of $\mu$ ($\nu$), so that in the end we obtain the sets $\delta(q_0, \mu)$ and $\delta(q_0, \nu)$.
We proceed similarly for $\gamma$, but this time we compute the set $\delta(q, \gamma)$ for each state $q \in Q$. 
Since $\mathcal D$ is the determinized version of $\mathcal A$, we can verify condition 2 of Proposition \ref{general polynomialW} by checking whether the set $\delta(q_0, \mu)$ and the set 
\[
\delta(q_0, \mu\cdot\gamma)=\bigcup_{p\in \delta(q_0, \mu)}\delta(p, \gamma)
\]
are equal, and we do the same for $\delta(q_0, \nu)$ and $\delta(q_0, \nu\cdot\gamma)$.
Condition 3 of Proposition \ref{general polynomialW} can be checked in constant space. To confront $\mu$ and $\gamma$, we use a variable $\rho$ that indicates whether $\mu$ is less, equal or greater than $\gamma$.
We initialize $\rho$ based on the counters $u,g$ as follows:
\[
\rho:=\begin{cases}
= \quad &\text{if }u=g\\
\vdash \quad &\text{if }u<g.
\end{cases}
\]
We leave $\rho$ unchanged until we start guessing simultaneously $\mu$ and $\gamma$. 
Then, when we guess simultaneously the character $c_1$ for $\mu$ and the character $c_2$ for $\gamma$, we set 
\[
\rho:=\begin{cases}
\prec \quad &\text{if }c_1\prec c_2\\
\succ \quad &\text{if }c_1\succ c_2\\
\rho \quad &\text{if }c_1=c_2.
\end{cases}
\]
Note that if at the end $\rho$ has value $\vdash$, it means that $\mu \vdash \gamma$. Otherwise, we have $\mu \,\rho\, \gamma$. Therefore, we are always able to determine the co-lexicographic order of $\mu$ and $\gamma$. 
To check condition 1 of Proposition \ref{general polynomialW}, consider the automata $A_{\mu}$ and $A_{\nu}$ obtained from the NFA $\mathcal A$ by considering as initial states the sets  $\delta(q_0, \mu)$ and  $\delta(q_0, \nu)$, respectively. 
We have that $\mu \not\equiv_{\mathcal L} \nu$ if and only if $\la{A_{\mu}} \neq \la{A_{\nu}}$, and checking whether $\la{A_{\mu}} = \la{A_{\nu}}$ can be done in polynomial space, since deciding whether two NFA's recognize the same language is a well-known PSPACE-complete problem.

To prove the completeness of the problem, we will show a polynomial reduction from the universality problem for NFA, i.e. the problem of deciding whether the language accepted by an NFA $\mathcal A$, over the alphabet $\Sigma$, is $\Sigma^*$. 

Let $\mathcal A = (Q, q_0,\delta, F, \Sigma)$ be an NFA and let $\mathcal L = \la A$. We can assume without loss of generality that $q_0 \in F$, otherwise $\mathcal A$ would not accept the empty string and we could immediately derive that $\mathcal L \ne \Sigma^*$. Let $a,b,c$ be three characters not in $\Sigma$ and such that $a\prec b\prec c$ with respect to the lexicographical order (the order of the characters of $\Sigma$ is irrelevant in this proof). First, we build the automaton $\mathcal A'$ starting from $\mathcal A$ by adding an edge $(q_f, q_0, c)$ for each final state $q_f \in F$, see the top part of Figure \ref{fig:pspace}. Notice that $\mathcal A'$ recognizes the language $\mathcal L' = \la{A'} = (\mathcal Lc)^* \cdot \mathcal L$, and it is straightforward to prove that $\mathcal L = \Sigma^*$ if and only if $\mathcal L' = (\Sigma + c)^*$: if $\mathcal L = \Sigma^*$, let $\alpha$ be a string in $(\Sigma + c)^*$ containing $n$ occurrences of $c$. Then $\alpha = \alpha_0\, c\, \alpha_2\, c\,\dots\, \alpha_{n-1}\, c\, \alpha_{n}$ for some $\alpha_1, \dots, \alpha_{n} \in \Sigma^*$. Hence $\alpha \in (\Sigma^*c)^*\cdot \Sigma^* = \mathcal L'$. On the other hand, if $\mathcal L \ne \Sigma^*$ let $\alpha$ be a string in $\Sigma^* \setminus \mathcal L$. Then $\alpha \cdot c \notin \mathcal L'$.

\begin{figure}
\begin{center}
\begin{tikzpicture}[->,>=stealth', semithick, initial text={}, auto, scale=.4]
\node[state, label=above:{},initial] (0) at (-6,0) {$q'_0$};

\node[state, label=above:{}, accepting] (1) at (0,5) {$q_0$};
\node[state, label=above:{}, accepting] (2) at (5,-3) {$q_1$};
\node[state, label=above:{}] (3) at (10,5) {N};
\node[state, label=above:{}, accepting] (5) at (15,5) {A};
\node[] (6) at (5,5) {$\mathcal A$};
\node (7) at (-4.5,8) {$\mathcal A'$};
\draw[dashed, rounded corners = 10] (-2,3.3) rectangle (17,6.8);
\draw[dashed, rounded corners = 10] (-3,2.3) rectangle (18,9.8);

\draw (0) edge node {$a$} (1);
\draw (0) edge[below] node {$b$} (2);

\draw (1) edge[loop above] node {c} (1);

\draw (5) edge[bend right = 40, above] node {$c$} (1);
\draw (2) edge[loop right] node {$\Sigma, c$} (2);

\end{tikzpicture}
\end{center}
\caption{The automaton $\mathcal A''$. Every accepting state of $\mathcal A$, labeled A in the figure, has a back edge labeled $c$ connecting it to $q_0$. Conversely, non-accepting states of $\mathcal A$, labeled N in the figure, do not have such back edges.}
\label{fig:pspace}
\end{figure}
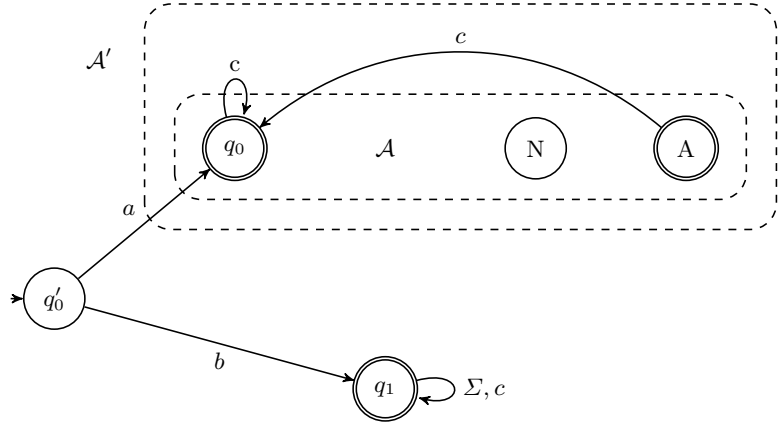

We build a second automaton $\mathcal A''$ as depicted in Figure \ref{fig:pspace}. Let $\mathcal L'' = \la {A''}$ be the language recognized by $\mathcal A''$. We claim that $\mathcal L = \Sigma^*$ if and only if $\mathcal L''$ is Wheeler.
\\$(\Longrightarrow)$ If $\mathcal L = \Sigma^*$, we have already proved that $\mathcal L' = (\Sigma+c)^*$. Hence we have $\mathcal L'' = (a + b) \cdot (\Sigma+c)^*$. The minimum DFA recognizing $\mathcal L''$ has only one loop, therefore by Theorem \ref{polynomialW} $\mathcal L''$ is Wheeler.
\\$(\Longleftarrow)$ If $\mathcal L \ne \Sigma^*$, let $\alpha$ be a string in $\Sigma^* \setminus \mathcal L$. Note that $\alpha \ne \varepsilon$ since we assumed that $\varepsilon \in \mathcal L$. Every possible run of $\alpha$ over $\mathcal A$ must lead to a non-accepting state, hence $\alpha \cdot c \notin \mathcal L'$. 
This implies that for all $i \ge 0$ we have $a \cdot c^i \cdot \alpha \cdot c \notin \mathcal L''$ (notice that the only edge labeled $c$ leaving $q_0$ ends in $q_0$). On the other hand, for all $j \ge 0$ we have $bc^j \cdot \alpha \cdot c \in \mathcal L''$, hence for all $i, j \ge 0$ we have $ac^i \not\equiv_{\mathcal L''} bc^j$.
Thus the following monotone sequence in $\pf {L''}$
\[
ac \prec bc \prec acc \prec bcc \prec \dots \prec ac^n \prec bc^n \prec \dots
\]
is not eventually constant modulo $\equiv_\mathcal{L''}$. From Lemma \ref{monotone} it follows that $\mathcal L''$ is not Wheeler.

Note that in the reduction described in   Figure \ref{fig:pspace}, if the starting NFA $\mt A$ was reduced, then also $\mt A''$ would be reduced. 
This means that the statement of the theorem   holds even if restricted to reduced NFA's.
\end{proof}

\begin{remark}
Note that  the previous theorem is in contrast with what happens when we consider the problem of deciding whether an NFA is Wheeler, instead of whether it accepts a Wheeler language: in that case, restricting the problem to reduced NFA's makes it solvable in polynomial time.
\end{remark}

\section{State complexity}
As already mentioned above, a significant property on the interplay between deterministic and non-deterministic Wheeler Automata is that given a size-$n$ WNFA $\mathcal A$, there always exists a WDFA that recognizes the same language whose size is at most $2n$. 
The announced amount of states can be computed using the (classic) powerset construction. 
In other words, the  blow-up of the number of states that we might observe  when converting NFA's to DFA's, does not occur for Wheeler non-deterministic automata.
This property is a  direct consequence of an important feature of Wheeler automata: for any state $q$, the set of strings recognized by $q$---namely $I_q$---is an interval over $\pf L$ with respect to the co-lexicographic order.  

State complexity is also used to measure the complexity of operations on regular languages. In the next section we  prove that the  interval property of a Wheeler DFA can  also be exploited to prove that    the state complexity of the intersection of Wheeler languages is
significantly better than the state complexity of the intersection of general regular languages. 

\subsection{Intersecting Wheeler languages}

The state complexity of a regular language $\mt L$ is defined as the number of states of the minimum DFA $\mt D_\mt L$ recognizing $\mt L$.
The state complexity of an operation on regular languages is a function that associates to the state complexities of the operand languages the worst-case state complexity of the language resulting from the operation.
For instance, we say that the state complexity of the intersection of $\mt L_1$ and $\mt L_2$ is $mn$, where $m$ and $n$ are the number of states of $\mt D_{\mt L_1}$ and $\mt D_{\mt L_2}$ respectively. 
The bound $mn$ for the intersection can easily be proved using the state-product construction for $\mt D_{\mt L_1}$ and $\mt D_{\mt L_2}$, and it is a known fact that this bound is tight \cite{Yu1994TheSC}.

It is natural to define the Wheeler state complexity of a Wheeler language $\mt L$ as the number of states of the minimum WDFA $\mt D^W_\mt L$ recognizing $\mt L$.
In the following theorem, we show what it is the Wheeler state complexity of the intersection of two Wheeler languages $\mt L_1$ and $\mt L_2$.

\begin{theorem}
\label{intersection}
Let $\mathcal D^W_{\mathcal L_1}$ and $\mathcal D^W_{\mathcal L_2}$ be the minimum WDFA's recognizing the languages $\mathcal L_1$ and $\mathcal L_2$ respectively. Then, the minimum WDFA recognizing $\mathcal L:=\mathcal L_1\cap \mathcal L_2$ has at most $|D^W_{\mathcal L_1}|+|D^W_{\mathcal L_2}|-|\Sigma|-1$ states.

\noindent This bound is tight.
\end{theorem}
\begin{proof}
First we prove that, given any two strings $\alpha,\beta \in \Sigma^*$, if $\alpha \equiv^c_{\mathcal L_1} \beta$ and $\alpha \equiv^c_{\mathcal L_2} \beta$ then $\alpha \equiv^c_{\mathcal L} \beta$. 
From $\alpha \equiv_{\mathcal L_1}\beta$ and $\alpha \equiv_{\mathcal L_2}\beta$ it follows that $\alpha \equiv_{\mathcal L}\beta$. 
Moreover, from $\alpha \equiv^c_{\mathcal L_1} \beta$ it follows that $\alpha$ and $\beta$ end with the same letter.
What it is left to prove is that for any $\gamma \in \Sigma^*$ such that $\alpha \prec \gamma \prec \beta$ it holds $\alpha \equiv_{\mathcal L} \gamma$. This follows immediately since $\alpha \equiv^c_{\mathcal L_1} \beta$ implies $\alpha \equiv_{\mathcal L_1} \gamma$ and $\alpha \equiv^c_{\mathcal L_2} \beta$ implies $\alpha \equiv_{\mathcal L_1} \gamma$.

Let $C^1_0, \dots C^1_{n-1}$ be the $\equiv^c_{\mt L_1}$-classes and let $C^2_0, \dots C^2_{m-1}$ be the $\equiv^c_{\mt L_2}$-classes; we assume that both lists are ordered co-lexicographically.
Since the $\equiv^c_{\mt L_1}$-classes are pairwise disjoint---and the same holds for the $\equiv^c_{\mt L_2}$-classes---the number of $\equiv_{\mathcal L_1\cap \mathcal L_2}^c$-classes is at most equal to the number of non-empty intersections of the form $C^1_i\cap C^2_j$, for $1\le i\le n$ and $1\le j\le m$. 
Classes that end with different characters of the alphabet must have empty intersection; a particular case are the classes $C^1_0=C^2_0=\{\varepsilon\}$, which always lead to the non-empty intersection $C^1_0\cap C^2_0 = \{\varepsilon\}$.
We will focus on classes whose elements end with a specific character, say $a$. 
Let $C^{1a}_1,\dots, C^{1a}_{n_a}$ be all the $\equiv^c_{\mt L_1}$-classes that end with $a$, co-lexicographically ordered, and let let $C^{2a}_1,\dots, C^{2a}_{m_a}$ be all the $\equiv^c_{\mt L_2}$-classes that end with $a$.
Let $k$ be the number of non-empty intersections of the form $C^{1a}_i\cap C^{2a}_j$, and let $\alpha_1 \prec \dots \prec \alpha_k$ be an ordered list containing one representatives for each non-empty intersection.
For any $1 \le s < k$, consider the strings $\alpha_s$ and $\alpha_{s+1}$. There must exist four unique indexes $i,j,i',j'$ such that $\alpha_s \in C^{1a}_i \cap C^{2a}_j$ and $\alpha_{s+1} \in C^{1a}_{i'} \cap C^{2a}_{j'}$. 
From $\alpha_s \prec \alpha_{s+1}$ it follows that both $i \le i'$ and $j \le j'$ hold, since the $\equiv^c_{\mt L_1}$-classes---and the $\equiv^c_{\mt L_2}$-classes---are pairwise disjoint and co-lexicographically ordered. 
On the other hand, it can not be the case that both $i=i'$ and $j=j'$ hold, because $\alpha_s$ and $\alpha_{s+1}$ belong to different intersections. 
Therefore we have that $i'+j' \ge i+j+1$.
The values of the function $f(\alpha_s)=i+j$ can range from 2 to $n_a+m_a$, hence there might be at most $n_a+m_a-1$ different representatives. 
Taking the sum over every possible characters of $\Sigma$ and adding the class $C^1_0\cap C^2_j = \{\varepsilon\}$, we get an upper bound of
\begin{align*}
1+\sum_{a\in\Sigma}(n_a+m_a-1)&=1+\sum_{a\in\Sigma}n_a+\sum_{a\in\Sigma}m_a-|\Sigma|=\\
&=1+(n-1)+(m-1)-|\Sigma|=n+m-|\Sigma|-1
\end{align*}
different possible representatives.

To show that the bound is tight (at least for $|\Sigma|=2$), consider the following families of languages over the alphabet $\Sigma=\{a,b\}$, with  $a\prec b$:
\begin{align*}
    &A_n:=\{\alpha\in\Sigma^*: \; a^{n+1}\text{ is not a factor of } \alpha\} \\
    &B_m:=\{\beta\in\Sigma^*: \; b^{m+1}\text{ is not a factor of } \beta\}.
\end{align*}
We can easily prove that all these languages are Wheeler. 
The minimum DFA recognizing $B_m$ is already a WDFA, see Figure \ref{Bm}, with $m+2$ states. 
A list of representatives of such classes is
\[
\varepsilon, a, b, \dots, b^m.
\]

\begin{figure}[H]
\begin{center}
\begin{tikzpicture}[->,>=stealth', semithick, initial text={}, auto, scale=.6]
\node[state, label=above:{},initial, accepting] (0) at (0,3) {$q_0$};
\node[state, label=above:{}, accepting] (1) at (8,6) {$q_1$};
\node[state, label=above:{}, accepting] (2) at (3,0) {$q_2$};
\node[state, label=above:{}, accepting] (3) at (8,0) {$q_3$};
\node[state, label=above:{}, accepting] (4) at (13,0) {$q_4$};

\draw (0) edge node {$a$} (1);
\draw (0) edge[below] node {$b$} (2);
\draw (2) edge[below] node {$b$} (3);
\draw (3) edge[below] node {$b$} (4);
\draw (1) edge[loop above] node {$a$} (1);
\draw (3) edge[below] node[xshift=8pt] {$a$} (1);
\draw (4) edge[below] node {$a$} (1);

\draw (2) edge[bend right = 10, below] node {$a$} (1);
\draw (1) edge[bend right = 10, above] node {$b$} (2);

\end{tikzpicture}
\end{center}
\caption{The minimum WDFA recognizing $B_3$.}
\label{Bm}
\end{figure}
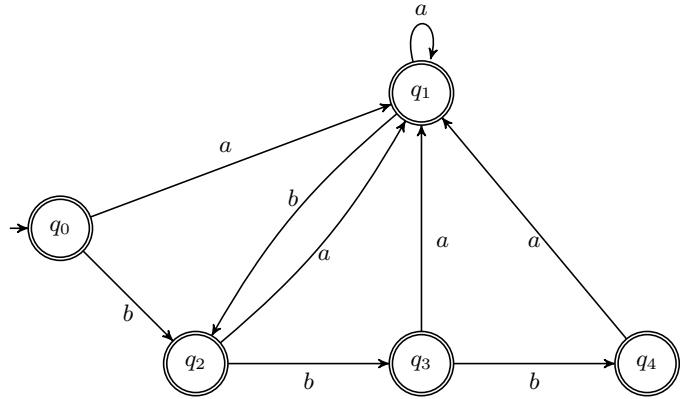

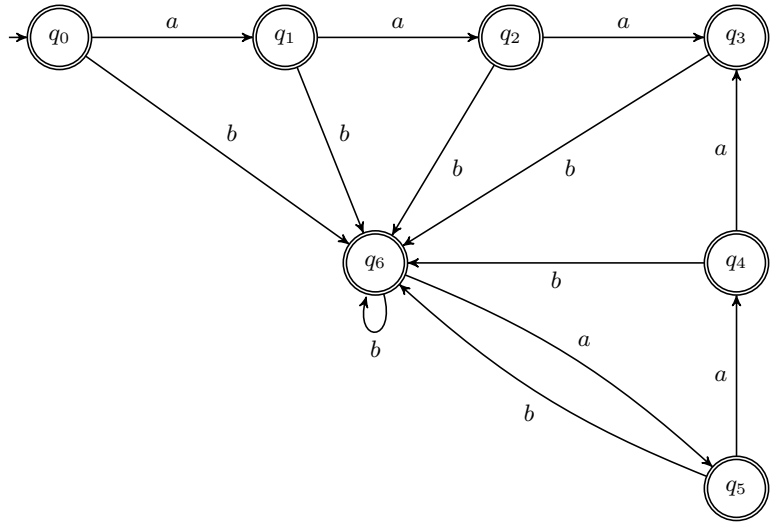
\begin{figure}
\begin{center}
\begin{tikzpicture}[->,>=stealth', semithick, initial text={}, auto, scale=.6]
\node[state, label=above:{},initial, accepting] (0) at (0,10) {$q_0$};
\node[state, label=above:{}, accepting] (1) at (5,10) {$q_1$};
\node[state, label=above:{}, accepting] (2) at (10,10) {$q_2$};
\node[state, label=above:{}, accepting] (3) at (15,10) {$q_3$};
\node[state, label=above:{}, accepting] (4) at (15,5) {$q_4$};
\node[state, label=above:{}, accepting] (5) at (15,0) {$q_5$};
\node[state, label=above:{}, accepting] (6) at (7,5) {$q_6$};

\draw (0) edge node {$a$} (1);
\draw (1) edge node {$a$} (2);
\draw (2) edge node {$a$} (3);
\draw (5) edge node {$a$} (4);
\draw (4) edge node {$a$} (3);
\draw (6) edge[bend left = 10] node {$a$} (5);

\draw (0) edge node {$b$} (6);
\draw (1) edge node {$b$} (6);
\draw (2) edge node {$b$} (6);
\draw (3) edge node {$b$} (6);
\draw (4) edge node {$b$} (6);
\draw (5) edge[bend left = 10] node {$b$} (6);

\draw (6) edge[loop below] node {$b$} (1);

\end{tikzpicture}
\end{center}
\caption{The minimum WDFA recognizing $A_3$.}
\label{An}
\end{figure}

\noindent The minimum WDFA recognizing $A_n$ has more states than the minimum DFA: for $1\le i < n$ we have that $a^i \prec a^n \prec ba^i$, hence we have to split the $\equiv_{A_n}$-class containing both $a^i$ and $ba^i$ into two different $\equiv^c_{A_n}$-classes. The automaton has $2n+1$ states, see Figure \ref{An}. 
A list of representatives of the $\equiv_{A_n}^ c$-classes is
\[
\varepsilon, a, \dots, a^n, ba^{n-1}, \dots, ba, b.
\]

We have already proved that the language $\mathcal L:= A_n \cap B_m$ might have at most $(2n+1)+(m+2)-|\Sigma|-1=2n+m$ different $\equiv^c_{\mathcal L}$-classes, hence it is sufficient to show that there are at least $2n+m$ different ones.
We claim that the $2n+m$ strings
\[
\varepsilon, a, \dots, a^n, ba^{n-1}, \dots, ba, b,\dots, b^m
\]
all belong to different $\equiv^c_{\mathcal L}$-classes. 
Strings that ends with a different amount of $a'$s (or $b'$s) belong to different $\equiv_{\mathcal L}$-classes, so there is nothing to prove.
Therefore we only have to check, for each $1\le i < n$, that $a^i$ and $ba^i$ belong to different $\equiv^c_{\mathcal L}$-classes, and again this is true since $a^i \prec a^n \prec ba^i$.
\end{proof}

\begin{remark}
Similarly to the case of determinizing a WNFA, where we can use the classic powerset construction without generating too many states, to compute a WDFA that recognizes the intersection of the languages accepted by two WDFA's $\mathcal W_1$ and $\mathcal W_2$ we can use the classic state-product construction with the certainty that it will not produce more states that necessary; that is, the number of states generated will be at most the sum of the number of states of $\mathcal W_1$ and $\mathcal W_2$.
\end{remark}

\begin{opprob}
Wheeler automata are closed under few operations: intersection and right-concatenation with a finite language, i.e. if $\mathcal L$ is a Wheeler language and $\mathcal F$ is a finite language, then also $\mathcal L \cdot \mathcal F$ is a Wheeler language. 
In general, the state complexity of the concatenation of $\mathcal L(\mathcal D_1)\cdot \mathcal L(\mathcal D_2)$ can result in an exponential blow-up in the number of states of $\mathcal D_2$ \cite{Yu1994TheSC}, even when restricted to finite languages \cite{Campeanu}. 
It remains open the question whether it is possible to obtain a better---that is, sub-exponential---upper bound for Wheeler automata.
\end{opprob}

\subsection{Computing the minimum WDFA}
Despite of the good behaviour that Wheeler automata show regarding determinization and intersection, there are cases when the state complexity of a construction is exponential. In fact, it is known \cite{ADPP} that a blow-up of  states can occur when switching from the minimum DFA recognizing a language $\mathcal L$ to its minimum WDFA. 
As a last contribution we provide an algorithm to compute the minimum WDFA starting from the minimum DFA $\mathcal D_{\mathcal L}$ of a Wheeler language $\mathcal L$, consisting in two steps: first, we describe an algorithm that extracts a \emph{fingerprint} of $\mathcal L$ starting from $\mathcal D_{\mathcal L}$, that is, a set of string containing exactly one representative of each $\equiv_\mathcal L^c$-class of $\mathcal L$.
Second, we provide an algorithm that builds the minimum WDFA recognizing $\mathcal L$ starting from any of its fingerprints. 

\begin{definition}[Fingeprint]
Let $\mathcal L$ be a Wheeler language, and let $m$ be the number of equivalence classes of $\equiv^c_{\mathcal L}$. 
A set of strings $F=\{\alpha_1, \dots, \alpha_m\}\subseteq \Sigma^*$ is called a \emph{fingerprint} of $\mathcal L$ if and only if for each $\equiv^c_{\mathcal L}$-class $C$ it holds $|F\cap C|=1$. 
\end{definition}

We start by proving that we can impose an upper bound to the length of the representative of a fingerprint.

\begin{lemma}
\label{short}
Let $\mathcal D_\mt L$ be the minimum DFA recognizing the Wheeler language $\mathcal L$ over the alphabet $\Sigma$, and let $C_1,...,C_m$ be the pairwise distinct equivalence classes of $\equiv_\mathcal L^c$. Then, for each $1 \le i \le m$, there exists a string $\alpha_i \in C_i$ such that $|\alpha_i| < n + n^2$, where $n := |\mathcal D_\mt L|$.
\end{lemma}
\begin{proof}
Suppose by contradiction that there exists a class $C_i$ such that for all $\alpha \in C_i$ it holds $|\alpha| \ge n + n^2$, and let $\alpha \in C_i$ be a string of minimum length. Consider the first $n+1$ states $q_0=t_0, ..., t_n$ of $\mathcal D_\mt L$ visited by reading the first $n$ characters of $\alpha$. Since $\mathcal D_\mt L$ has only $n$ states, there must exist $0 \le i,j \le n$ with $i < j$ such that $t_i=t_j$. Let $\alpha'$ be the prefix of $\alpha$ of length $i$ (if $i=0$ then $\alpha' = \varepsilon$), let $\delta$ be the factor of $\alpha$ of length $j-i$ labeling the path $t_i,...,t_j$, and let $\zeta$ be the suffix of $\alpha$ such that $\alpha = \alpha' \delta \zeta$. By construction, the strings $\alpha$ and $\beta := \alpha' \zeta$ end in the same state, hence $\alpha \equiv_\mathcal L \beta$. Moreover, from $|\beta| < |\alpha|$ and the minimality of $\alpha$ it follows that $\alpha \not\equiv_\mathcal L^c \beta$. 
\\Suppose that $\alpha \prec \beta$, the other case being completely symmetrical. Since $\alpha$ and $\beta$ share the same suffix $\zeta$, they end with the same character. This means that the strings $\alpha$ and $\beta$, which are Myhill-Nerode equivalent but not $\equiv_\mathcal L^c$ equivalent, were not split into two distinct $\equiv_\mathcal L^c$-classes due to input-consistency, therefore there must exists a string $\eta$ such that $\alpha \prec \eta \prec \beta$ and $\eta \not\equiv_\mathcal L \alpha$. Formally, assume by contradiction that for all strings $\eta$ such that $\alpha \prec \eta \prec \beta$ it holds $\eta \equiv_\mathcal L \alpha$. Then, by definition of $\equiv_\mathcal L^c$, it would follow $\alpha \equiv_\mathcal L^c \beta$, a contradiction.
\\Let $\eta$ be a string such that $\alpha \prec \eta \prec \beta$ and $\eta \not\equiv_\mathcal L \alpha$. From $\zeta \dashv \alpha, \beta$ it follows that $\zeta \dashv \eta$, so we can write $\eta = \eta' \zeta$ for some $\eta' \in \Sigma^*$. Recall that by construction $\alpha = \alpha' \delta \zeta$ with $|\alpha' \delta| \le n$, hence $|\zeta| \ge n^2$. Consider the last $n^2+1$ states $r_0, ..., r_{n^2}$ of $\mathcal D_\mt L$ visited by reading the string $\alpha$, and the last $n^2+1$ states $p_0, ..., p_{n^2}$ visited by reading the string $\eta$. Since $\mathcal D_\mt L$ has only $n$ states, there must exist $0 \le i,j \le n^2$ with $i < j$ such that $(r_i, p_i) = (r_j, p_j)$. Notice that it can't be $r_i = p_i$, otherwise from the determinism of $\mathcal D_\mt L$ it would follow $r_{n^2} = p_{n^2}$; from the minimality of $\mathcal D_\mt L$ it would then follow $\alpha \equiv_\mathcal L \eta$, a contradiction.
\\Let $\zeta''$ be the suffix of $\zeta$ of length $n^2-j$, and let $\gamma$ be the factor of $\zeta$ of length $j-i$ labeling the path $r_i,...,r_j$. Since $|\zeta| \ge n^2$, there exists $\zeta' \in \Sigma^*$ such that $\zeta = \zeta' \gamma \zeta''$. We can then rewrite $\alpha, \eta$ and $\beta$ as
\begin{align*}
    \alpha &= \alpha' \delta \zeta = \alpha' \delta \zeta' \gamma \zeta'' \\
    \eta &= \eta' \zeta = \eta' \zeta' \gamma \zeta'' \\
    \beta &= \alpha' \zeta = \alpha' \zeta' \gamma \zeta''.
\end{align*}
Let $k$ be an integer such that $|\gamma^k|$ is greater than $|\alpha' \delta \zeta'|$ and $|\eta' \zeta'|$. Set $\mu := \eta' \zeta'$; from $\alpha \prec \eta \prec \beta$ it follows that $\alpha' \delta \zeta' \prec \mu \prec \alpha' \zeta'$. If $\gamma^k \prec \mu$ set $\nu := \alpha' \zeta'$, otherwise set $\nu := \alpha' \delta \zeta'$. In both cases, the hypothesis of Theorem \ref{polynomialW} are satisfied, since $\gamma^k$ labels two cycles starting from the states $r_i$ and $p_i$, that we have proved to be distinct. We can conclude that $\mathcal L$ is not Wheeler, a contradiction, and the thesis follows. 
\end{proof}

We show now how to compute the minimum WDFA recognizing a Wheeler language $\mathcal L$ if we are given its minimum DFA $\mt D_\mt L$ and one of its fingeprints.

\begin{proposition}[Fingerprint to min WDFA]
\label{DFA to WDFA}
Let $\mathcal D_{\mathcal L}$ be the minimum automaton recognizing the Wheeler language $\mathcal L$ with $| \mathcal D_\mt L | = n$ and let $C_1,...,C_m$ be the pairwise distinct equivalence classes of $\equiv_\mathcal L^c$.
Assume that we are given a \emph{fingerprint} of $\mathcal L$, whose elements have length less than $n^2+n$.
Then it is possible to build the minimum WDFA recognizing $\mathcal L$ in $O(n^2\cdot\sigma\cdot m\log m)$ time.
\end{proposition}
\begin{proof}
Let $\{ \alpha_1, ..., \alpha_m \}$ be a fingerprint of $\mathcal L$ and let $\mathcal D_\mt L$ be the minimum DFA recognizing $\mathcal L$. We can assume without loss of generality that $\alpha_1 \prec ... \prec \alpha_m$. We build the automaton $\mathcal D_\mt L^W=(Q, \alpha_1, \delta, F,\Sigma)$, where the set of states is $Q = \{ \alpha_1, ..., \alpha_m \}$ and the set of final states is $F = \{ \alpha_j:\; \alpha_j \in \mathcal L \}$. The transition function $\delta$ can be computed as follow. For all $1 \le j \le m$ and for all $c \in \Sigma$, check whether $\alpha_j \cdot c \in \pf L$. If $\alpha_j \cdot c \notin \pf L$, there are no edges labeled $c$ that exit from $\alpha_j$. If instead $\alpha_j \cdot c \in \pf L$, in order to define $\delta(\alpha_j, c)$ we just have   to determine the $\equiv_{\mathcal L}^c$-class of the string $\alpha_j\cdot c$ (see Theorem \ref{wdeterminization}). We first  locate the position of $\alpha_j \cdot c$ in the intervals defined by  $\alpha_1 \prec ... \prec \alpha_m$ using a binary search. There are three possible cases.
\begin{enumerate}
    \item $\alpha_j \cdot c \preceq  \alpha_1$. Then by the properties of $\equiv_{\mathcal L}^c $ it easily follows $\alpha_j \cdot c \equiv_{\mathcal L}^c \alpha_1$  and we define  $\delta(\alpha_j, c) = \alpha_1$.
\item $\alpha_m \preceq \alpha_j \cdot c$. Similarly to the previous case, we have  $\alpha_j \cdot c \equiv_{\mathcal L}^c \alpha_m$  and we define $\delta(\alpha_j, c) = \alpha_m$.
\item There exists $s$ such that $\alpha_s \preceq \alpha_j \cdot c \preceq \alpha_{s+1}$. It can not be the case that both $\alpha_jc \not\equiv_\mathcal L \alpha_s$ and $\alpha_jc \not\equiv_\mathcal L \alpha_{s+1}$, since $\{ \alpha_1, ..., \alpha_m \}$ is a fingerprint of $\mathcal L$ and $\equiv_{\mathcal L}^c$-classes are intervals in $\pf L$. Hence we distinguish three cases.
\begin{enumerate}
    \item $\alpha_s \equiv_\mathcal L \alpha_j\cdot c \not \equiv_\mathcal L \alpha_{s+1}$. Then 
    $\alpha_j \cdot c \equiv_{\mathcal L}^c \alpha_s$ and we define $\delta(\alpha_j, c) = \alpha_s$.
    \item $\alpha_s \not\equiv_\mathcal L \alpha_j\cdot c \equiv_\mathcal L \alpha_{s+1}$. Then $\delta(\alpha_j, c) = \alpha_{s+1}$.
    \item $\alpha_s \equiv_\mathcal L \alpha_j\cdot c \equiv_\mathcal L \alpha_{s+1}$. Since $\{ \alpha_1, ..., \alpha_m \}$ is a fingerprint of $\mathcal L$, it is either $c = \text{end}(\alpha_jc) = \text{end}(\alpha_s)$, in which case $\alpha_j \cdot c \equiv_{\mathcal L}^c \alpha_s$ and we define $\delta(\alpha_j, c) = \alpha_s$, or $c = \text{end}(\alpha_{s+1})$, in which case $\alpha_j \cdot c \equiv_{\mathcal L}^c \alpha_{s+1}$ and we define $\delta(\alpha_j, c) = \alpha_{s+1}$ (where by $\text{end}(\beta)$ we denote the last letter of the string $\beta$, for $\beta \in \Sigma^+$).
\end{enumerate}
\end{enumerate}
\end{proof}

To complete the construction, we show how to extract a \emph{fingerprint} of a Wheeler language $\mathcal L$ starting from its minimum DFA.
We first need to prove the following Lemma.

\begin{lemma}
\label{bounded string}
Given a DFA $\mathcal D$ with $n$ states, a state $q$ and a string $\gamma \notin I_q$ with $|\gamma|\le n^2+n$, we can find in polynomial time, if it exists, the greatest (smallest) string in $I_q$ that is smaller (greater) than $\gamma$ and has length at most $n^2+n$.
\end{lemma}
\begin{proof}
Let UB the the upper bound $\text{UB}=n^2+n$. Using dynamic programming, we can extract a $n\times$UB table storing, for each $(i,j)$, the smallest and the greatest string in $I_{q_i}$ of length at most $j$ (see [ADPP]).
Given a string $\alpha$, we use the notation $\alpha[i]$ to denote the $i$-th to last character of $\alpha$ (or $\varepsilon$ if $i>|\alpha|$),  and  the notation $\alpha_i$ to denote the suffix of $\alpha$ of length $i$. In particular we have $\alpha_{i+1}=\alpha[i+1]\cdot\alpha_i$.
In this Lemma we are interested only in strings with length less than UB, therefore every string (subset of strings) that will be mentioned has to be intended as an element (subset, respectively) of $\Sigma^{\le\text{UB}} = \{\alpha \in \Sigma^*: \; |\alpha|\le \text{UB}\}$.

We want to find the greatest string in $I_q$ that is smaller than $\gamma$. 
Note that if $\gamma$ is the suffix of a string $\alpha$, then $\gamma \prec \alpha$ so we do not have to worry about strings ending with $\gamma$.
Note also that the greatest string smaller than $\gamma$ must maximize the length of the longest suffix it has in common with $\gamma$.
Therefore, we look for all the states of $\mt D$ starting from which it is possible to read the longest proper suffix of $\gamma$ that ends in $q$. 
To do that, for each $1\le i < |\gamma|$ we build the set $S_i=\{p\in Q:\; p \overset{\gamma_i}{\leadsto} q\}$. 
We start from the set $S_0=\{q\}$ and to build $S_{i+1}$ from $S_i$ we simply follow backward the edges labeled $\gamma[i+1]$.
Every time we determine a set $S_i$, we check if there exists at least one incoming edge with a label strictly less than $\gamma[i+1]$. 
If this is the case, we keep in memory $S_i$ as the last set we built with such property; previously stored sets can be overwritten. 
This procedure ends either when we find an $S_i$ that is empty or when we successfully build the last set $S_{|\gamma|-1}$. 
If we did not store any of the $S_i$ we have built, then there is no string in $I_q$ smaller than $\gamma$.
If instead we have stored at least one $S_i$, we consider the last one stored (that is, the only one that has not been overwritten), say $S_k$. 
Clearly, the computation of any string in $I_q$ smaller than $\gamma$ that maximizes the length of the longest suffix it has in common with $\gamma$  must reach a state of $S_k$ at his $k$-th to last step. 
Therefore, let $c$ be the greatest label smaller than $\gamma[k+1]$ that enters $S_k$ (note that $c$ must exists since we stored $S_k$), and let $S$ be the set of states that can reach $S_k$ by an edge labeled $c$. 
Using the table computed at the very beginning of this lemma, we can easily find, if it exists, the greatest string $\bar\alpha$ of length at most UB$-(k+1)$ that can reach a state of $S$. 
Then, the greatest string in $I_q$ that is smaller than $\gamma$ is $\bar\alpha\cdot c \cdot \gamma_k$.

To find the smallest string in $I_q$ that is greater than $\gamma$, we split the problem into two sub-problems: 1) find the smallest string in $I_q$ that is greater than $\gamma$ but has not $\gamma$ as a suffix and 2) find the smallest string in $I_q$ that has $\gamma$ as a suffix.
The first problem is a symmetric version of the one discussed above, and can be solved in a similar way: we use exactly the same sets $S_i$, but this time we store a set $S_i$ if there exists at least one incoming edge with a label strictly greater than $\gamma[i+1]$.
To also solve the second problem, instead of stopping when computing $S_{|\gamma|-1}$ we carry on and compute $S_{|\gamma|}$. We do this since the following implication holds: there exists at least one string in $I_q$ that has $\gamma$ as a suffix iff $S_{|\gamma|}$ is not empty and there is at least one string of length at most $UB-|\gamma|$ that can reach a state of $S_{|\gamma|}$. 
If $S_{|\gamma|}\ne \emptyset$, we use again the table to determine, if it exists, the smallest string $\bar\beta$ of length at most $UB-|\gamma|$ that can reach a state of $S_{|\gamma|}$.
Lastly, we confront $\bar\beta\cdot\gamma$ with the string obtained by solving the first problem and we choose the smaller one.
\end{proof}

As a last step, Algorithm \ref{alg} generates a fingerprint of a language $\mt L$ starting from the minimum DFA $\mt D_\mt L$. 
The algorithm uses the subroutines described in Lemma \ref{bounded string}: given a DFA $\mt D$ with set of states $Q=\{q_0,\dots,q_{n-1}\}$ and two strings $m, m'\in \text{Pref}(\la D)$ with $m\in I_{q_k}$ (for some $0\le k \le n-1$),
\begin{itemize}
    \item MinMaxPair$(\mt D)$ returns the set of pairs $(m_0,M_0),\dots,(m_{n-1},M_{n-1})$, where $m_i$ is the co-lexicographically smallest string in $I_{q_i}$ of length at most $n^2+n$, and $M_i$ is the greatest.
    \item GreatestSmaller$(m, m', \mt D)$ returns the greatest string in $I_{q_k}$ smaller than $m'$ of length at most $n^2+n$.
    \item SmallestGreater$(m, m', \mt D)$ returns the smallest string in $I_{q_k}$ greater than $m'$ of length at most $n^2+n$.
\end{itemize}
\begin{algorithm}   
    \caption{Min DFA to fingerprint}\label{alg}
    \begin{algorithmic}[1]
        \Require{The minimum DFA $\mt D_\mt L$ recognizing $\mt L$}
        \Ensure{ A fingerprint of $\mt L$}
        \Statex
        \State $\tau \leftarrow$ MinMaxPairs$(\mt D_\mt L)$
        \Comment{We initialize a set of $|\mt D_\mt L|$ pairs of strings}
        \Statex
        \While{there exist $(m, M),(m', M')\in \tau$ such that $m \prec m' \prec M$}
            \State $M_1\leftarrow$ GreatestSmaller$(m, m', \mt D_\mt L)$
            \State $m_2\leftarrow$ SmallestGreater$(m, m', \mt D_\mt L)$
            \State $\tau\leftarrow\tau\setminus\{(m, M)\}$
            \State $\tau\leftarrow\tau\cup\{(m, M_1), (m_2, M)\}$
        \EndWhile
        \Statex
        \State $\tau\leftarrow$ Expand$(\tau)$
        \State \textbf{return} the first component of each element of $\tau$
    \end{algorithmic}
\end{algorithm}
At each iteration of the while cycle, we check the existence of two overlapping pairs $(m, M), (m', M')$ and replaces the first one with two new pairs $(m, M_1)$ and $(m_2, M)$. 
As we will prove in the Appendix, this cycle always ends. Clearly, when we exit the cycle $\tau$ can not contain overlapping pairs.
We will also prove that, at this point, each pair $(m,M)\in \tau$ satisfy the following properties:
\begin{enumerate}
    \item $m$ and $M$ belong to the same $\equiv_\mt L$-class;
    \item if there exists a Wheeler class $C$ such that $m\prec C\prec M$, then $C\subseteq [m]_{\equiv_\mt L}$.
\end{enumerate}
Lastly, we use the subroutine Expand to extract, from each pair $(m,M)\in \tau$, a representative of all the Wheeler classes $C$ (if any exists) such that $m\prec C\prec M$.
Since property 1-2 hold, if $\text{end}(m)=\text{end}(M)$ there are no Wheeler classes $C$ such that $m\prec C\prec M$; moreover, $m$ and $M$ belong to the same Wheeler class, so we leave the pair $(m, M)$ unchanged.
Otherwise, if $\text{end}(m)\neq \text{end}(M)$ and there is  a Wheeler class $C$ is such that $m\prec C\prec M$, it must be the case that  the strings in $C$ end with a character that differs from both the last character of $m$ and the last one of $M$. For each character $c$ such that $\text{end}(m)\prec c \prec \text{end}(M)$, we check whether there exists a string $\alpha_c\in I_{\delta(q_0,m)}$ such that $\text{end}(\alpha_c)=c$.
Every time we find an $\alpha_c$ with such property, we add to $\tau$ the pair $(\alpha_c,\alpha_c)$. 
As a last step, we replace the pair $(m, M)$ with the pairs $(m, m)$ and $(M,M)$, since from $\text{end}(m)\neq \text{end}(M)$ it follows that $m$ and $M$ belong to different Wheeler classes.

After the Expand subroutine has been run, $\tau$ will contain exactly one pair for each Wheeler class $C$ of $\mt L$, whose components both belong to $C$.
By extracting from each pair one of its components, e.g. the first one, we obtain a fingerprint of $\mt L$.

\section{Conclusions}

In this paper we considered a number of computational complexity problems related with the general idea of  ordering states of a finite automaton. In general, ordering objects might  lead to significant simplification of otherwise difficult storage and/or manipulation problems. In fact, ordered finite automata  can ease such tasks as index construction, membership testing, and even determinization of NFA's accepting a given regular language. Clearly, a key point is the complexity of \emph{finding} the right order from scratch. Even though this turned out to be simple on DFA's and, as opposed to the non-ordered case, turning a Wheeler NFA into a Wheeler DFA is polynomial, things become much more tricky when the input automaton is a non-deterministic one. This issue, together with some of its natural variants, were the main theme of this paper. We proved that a number of ordered-related results, ultimately guaranteeing the existence of polynomial time algorithms on  DFA's, are much more complex if the starting automaton is an NFA---even in the case of a reduced NFA. 

The complexity bounds we studied and presented here suggest the ``dangerous'' directions along which generalisations can be searched. 

An interesting theme we did not explore here, is the possibility of exploiting  order over more general classes of automata and languages. Can ordering states of a push-down (deterministic) automata or even a (deterministic) Turing Machine, be a way to obtain a simplification of interesting problems  over the recognized languages?
Can we define an order over the states of a DFA and use this order to simplify problems relative to 
language over infinite strings?

The order imposed on states of an accepting automaton is reflecting, in a variety of ways, the underlying properties of the ordering on strings reaching that state. The co-lexicographic order seems to be an especially effective one. However, exploring this relationship---and the corresponding complexity bounds---in interesting and expressive contexts, can be extremely stimulating in terms challenging formal language problems.

\bibliographystyle{splncs04}
\bibliography{bibliography}

\begin{thebibliography}{1}
\providecommand{\url}[1]{\texttt{#1}}
\providecommand{\urlprefix}{URL }
\providecommand{\doi}[1]{https://doi.org/#1}

\bibitem{ADPP}
Alanko, J., D'Agostino, G., Policriti, A., Prezza, N.: Regular languages meet
  prefix sorting. In: Proceedings of the 2020 ACM-SIAM Symposium on Discrete
  Algorithms. pp. 911--930 (2020). \doi{10.1137/1.9781611975994.55},
  \url{https://epubs.siam.org/doi/abs/10.1137/1.9781611975994.55}

\bibitem{ADPP2}
Alanko, J., D'Agostino, G., Policriti, A., Prezza, N.: Wheeler languages. Inf.
  Comput.  \textbf{281}(C) (dec 2021). \doi{10.1016/j.ic.2021.104820},
  \url{https://doi.org/10.1016/j.ic.2021.104820}

\bibitem{Campeanu}
C{\^a}mpeanu, C., Culik, K., Salomaa, K., Yu, S.: State complexity of basic
  operations on finite languages. In: Boldt, O., J{\"u}rgensen, H. (eds.)
  Automata Implementation. pp. 60--70. Springer Berlin Heidelberg, Berlin,
  Heidelberg (2001)

\bibitem{JACM}
Cotumaccio, N., D'Agostino, G., Policriti, A., Prezza, N.: A theory of (co-lex)
  ordered regular languages. In preparation (2022)

\bibitem{NN}
Cotumaccio, N., Prezza, N.: On Indexing and Compressing Finite Automata, pp.
  2585--2599. \doi{10.1137/1.9781611976465.153},
  \url{https://epubs.siam.org/doi/abs/10.1137/1.9781611976465.153}

\bibitem{gibney2020simple}
Gibney, D., Hoppenworth, G., Thankachan, S.V.: Simple reductions from
  formula-sat to pattern matching on labeled graphs and subtree isomorphism.
  In: Le, H.V., King, V. (eds.) 4th Symposium on Simplicity in Algorithms,
  {SOSA} 2021, Virtual Conference, January 11-12, 2021. pp. 232--242. {SIAM}
  (2021). \doi{10.1137/1.9781611976496.26},
  \url{https://doi.org/10.1137/1.9781611976496.26}

\bibitem{NP}
Gibney, D., Thankachan, S.V.: On the hardness and inapproximability of
  recognizing wheeler graphs. In: 27th Annual European Symposium on Algorithms,
  {ESA} 2019, September 9-11, 2019, Munich/Garching, Germany. pp. 51:1--51:16
  (2019). \doi{10.4230/LIPIcs.ESA.2019.51},
  \url{https://doi.org/10.4230/LIPIcs.ESA.2019.51}

\bibitem{Gagie}
Travis Gagie{,} Giovanni Manzini~e Sir{\'e}n, J.: Wheeler graphs: A framework
  for bwt-based data structures. Theoretical computer science  \textbf{698},
  67--78 (2017)

\bibitem{Yu1994TheSC}
Yu, S., Zhuang, Q., Salomaa, K.: The state complexities of some basic
  operations on regular languages. Theor. Comput. Sci.  \textbf{125},  315--328
  (1994)

\end{thebibliography}

\newpage
\section{Appendix}

\begin{proof}[Proof of Proposition \ref{general polynomialW}]
Let $\mathcal D_\mt L = (\hat Q, \hat q_0, \hat \delta, \hat F, \Sigma)$ be the minimum DFA recognizing $\mathcal L$. Clearly $\mathcal D_\mt L$ has at most $n$ states.
\\($\Longleftarrow$) From condition 2 it follows that $\delta(q_0,\mu)=\delta(q_0,\mu\gamma)$, thus $\mu\equiv_\mt L\mu\gamma$. Therefore, in $\mt D_\mt L$ we also have $\delta(\hat q_0,\mu)=\delta(\hat q_0,\mu\gamma)$. Similarly, it holds $\delta(\hat q_0,\nu)=\delta(\hat q_0,\nu\gamma)$.
It follows that $\mu, \nu,\gamma$ satisfy condition 1-3 of Theorem \ref{polynomialW}, hence $\mt L$ is not Wheeler.
\\$(\Longrightarrow)$
Since $\mathcal L$ is not Wheeler, let $\hat\mu, \hat\nu, \hat\gamma$ be three strings satisfying conditions 1-4 of Theorem \ref{polynomialW}. The DFA $\mathcal D_\mt L$ has at most $n$ states, hence the length of $\hat\mu,\hat\nu$ and $\hat\gamma$ is bounded by $n^3+2n^2+n+2$.
We have $\hat\mu\hat\gamma^* \subseteq \pf L$, so let $t_0 = q_0, t_1, \dots, t_m$ be a run of $\hat\mu\hat\gamma^n$ over $\mathcal D$. We set $u := |\hat\mu|$ and $g := |\hat\gamma|$, and consider the list of $n+1$ states 
\[
t_u, \; t_{u+g}, \; t_{u+2g}, \; \dots, \; t_{u+ng} = t_m 
\]
Since $\mathcal D$ has $n$ states, there must exist two integers $0 \le h < k \le n$ such that $t_{u+hg} = t_{u+kg}$. That is, there exists a state $p := t_{u+hg}$ such that $p \in \delta\left(q_0, \hat\mu\hat\gamma^h\right)$ and $\hat\gamma^{k-h}$ labels a cycle starting from $p$. We can repeat the same argument for a run of $\hat\nu\hat\gamma^n$ over $\mathcal D$ to find a state $r$ and two integers $h', k'$ such that $r \in \delta(q_0, \hat\nu\hat\gamma^{h'})$ and $\hat\gamma^{k'-h'}$ labels a cycle starting from $r$. 
We define the constant $h''$ as the minimum multiple of $(k-h)\cdot(k'-h')$ greater than $\max\{h+1,h'+1\}$; it can be proved that $h''\le n^2$, and by construction $\hat\gamma^{h''}$ labels both a cycle starting from $p$ and one starting from $r$. 
We then define the strings
\begin{align*}
\mu &:= \hat\mu \hat\gamma^{h} \\
\nu &:= \hat\nu \hat\gamma^{h'} \\
\gamma &:= \hat\gamma^{h''}, \\
\end{align*}
which satisfy conditions 2 and 3. Note that we have chosen a $h''$ such that $|\gamma|>|\mu|, |\nu|$, so that $\gamma\ndashv\mu,\nu$.
Condition 4 is satisfied since $|\hat\gamma| \le n^3+2n^2+n+2$ and $h''\le n^2$.
Lastly, condition 1 is satisfied since the strings $\hat\mu$ and $\hat\mu\hat\gamma^h$ lead to the same state of $\mathcal D_\mt L$, thus $\hat\mu \equiv_\mathcal L \hat\mu\hat\gamma^h$. Similarly, we have $\hat\nu \equiv_\mathcal L \hat\nu\hat\gamma^{h'}$. 
The thesis then follows from $\hat\mu\not\equiv_\mt L\hat\nu$.
\end{proof}

\begin{proof}[Termination and correctness of Algorithm \ref{alg}]
We start by analyzing the subroutines used by the algorithm.
The subroutine MinMaxPairs can be computed simply by looking at the $n\times$UB table described in Lemma \ref{bounded string}.
Note that Lemma \ref{short} ensures that 
$m_{q_i}$ (respectively, $M_{q_i}$) belongs to the smaller (greatest) Wheeler class contained in $I_{q_i}$.
Subroutines GreatestSmaller and SmallestGreater are thoroughly described in Lemma \ref{bounded string}.

We now prove that Algorithm \ref{alg} always terminates.
Let $\tau_i$ be the set $\tau$ at the end of the $i-th$ iteration of the while cycle,
and for any pair of strings $c=(m, M)$ let $w(c)$ denote the number of Wheeler classes $C$ such that $m \prec C \prec M$ and $C\nsubseteq[m]_{\equiv_\mt L}$. 
Given a set of pairs $\tau$, let $w(\tau)$ denote the value
\[
w(\tau) := \sum_{c\in \tau} w(c)\ge0.
\]
We say that two pairs $c=(m, M)$ and $c'=(m', M')$ are \textit{ordered} if either $M \prec m'$ or $M' \prec m$.
In order to prove that $w(\tau_{i+1}) < w(\tau_i)$,
we will maintain the following invariants:
\begin{enumerate}
    \setcounter{enumi}{-1}
    \item if $c = (m, M)$ and $c' = (m', M')$ are two distinct pairs in $\tau_i$, then $\{m, M\}\cap \{m', M'\}=\emptyset$;
    \label{0}
    \item if $c = (m, M) \in \tau_i$, then $m \preceq M$,  $m \equiv_{\mathcal L} M$,  and end$(m)=$ end$(M)$; \label{1}
    \item if $c = (m, M)$ and $c' = (m', M')$ are two distinct pairs in $\tau_i$ such that $m \equiv_{\mathcal L} m'$, then $c$ and $c'$ are ordered. Moreover, if $c$ and $c'$ are also \emph{consecutive}, that is if there is no $c^*=(m^*, M^*)$ in $\tau_i$ such that $m \equiv_{\mathcal L} m^*$ and $m \prec m^* \prec m'$, then there is no Wheeler class $C \subseteq [m]_{\equiv_{\mathcal L}}$ such that $m \prec C \prec m'$; \label{2} 
    \item let $x$ be the first or second component of any pair $c$ in $\tau_i$  and $y$ be the first or second component of any pair $c'$ in $\tau_i$. If $x \equiv^c_{\mathcal L} y$, then $c = c'$. \label{3}
\end{enumerate}
Note that invariant \ref 0 implies that  every time we have two distinct, not ordered pairs $c=(m,M)$ and $c'=(m',M')$,   the   strict inequalities $m\prec M'$ and $m'\prec M$ hold. 
By construction, these invariants hold for $\tau_0$, which is the set returned by MinMaxPairs. 
For instance, invariant \ref{2}, \ref 3 hold since  in $\tau_0$  distinct pairs have components belonging to different $\equiv_\mathcal L$-classes.  

Invariants \ref{0}-\ref{2} can be easily proved by induction on $i$ just by looking at how new pairs are created.
We prove by induction invariant \ref{3}: suppose that it holds for $\tau_i$. 
Let $c = (m, M),c' = (m', M')$ be the pairs that meet the while condition on Line 2, and let $c_1 = (m, M_1), c_2 = (m_2, M)$ be the two pairs that replace $c$ on Line 5-6.
Note that $c,c'$ are not ordered, hence invariant 2 implies that $m$ and  $m'$ belong to different $\equiv_\mt L$-classes.
Suppose by contradiction that the invariant does not hold for $\tau_{i+1}$, that is, there exist two distinct pairs $d, d' \in \tau_{i+1}$ such that a component $x$ of $d$ belongs to the same Wheeler class of a component $y$ of $d'$. 
By induction, it can not be the case that $d,d'\in\tau_i$. 
Therefore, at least one among $d$ and $d'$ belongs to $\{c_1, c_2\}$. Moreover, $d$ and $d'$ can not both belong to $\{c_1, c_2\}$: we have by construction that the (possibly identical) Wheeler classes of $m$ and $M_1$ are different from the Wheeler classes of $m_2$ and $M$. 
We assume, w.l.o.g., that $d$ belongs to $\{c_1, c_2\}$ whereas $d'$ doesn't; in particular $d'\in \tau_i$ and since $d'\in \tau_{i+1}$ we also have $d'\neq c$.
There are two possibilities. If $d=c_1$, it can not be the case that $x=m$, otherwise the pairs $c, d' \in \tau_i$ would violate the inductive hypothesis. 
Thus $x=M_1$.
From $y \equiv^c_{\mathcal L} M_1 \equiv_{\mathcal L} m$ it follows that $y \equiv_{\mathcal L} m$  and invariants \ref{1}, \ref{2} applied to $\tau_i$ imply  that $c$ and $d'$ are ordered, that is, either $y \prec m$ or $M \prec y$ holds.
If $y \prec m$ we get $y \prec m \preceq M_1$, thus $m$ belongs to the same Wheeler class of $y$.
If $M \prec y$ we get $M_1 \preceq M \prec y$, thus $M$ belongs to the same Wheeler class of $y$.
In both cases, considering $c, d'\in\tau_i$, we reach a contradiction with our inductive hypothesis.
\\\noindent If instead $d=c_2$, we use a similar argument to show that $x=m_2$ and that either $y \prec m \preceq m_2$ or $m_2 \preceq M \prec y$ hold. Since both inequalities lead to a contradiction, we can conclude that invariant \ref{3} holds.
Hence we proved that invariant \ref 3 holds for all $\tau_i$.

We can now prove that $w(\tau_{i+1}) < w(\tau_i)$. Let $(m,M_1),(m_2,M)$ be the pairs added to $\tau_i$ on Line 6 of the Algorithm \ref{alg}.
Note that if $C$ is a Wheeler class such that $m \prec C \prec M$ and $C\nsubseteq [m]_{\equiv_\mt L}$, it can not be the case that both $m \prec C \prec M_1$ and $m_2 \prec C \prec M$ occurs, since $M_1 \prec m_2$. 
Moreover, let $C'$ be the Wheeler class containing $m'$.
From invariant 2 it follows that $C'\nsubseteq [m]_{\equiv_\mt L}$, and from $M_1 \prec m' \prec m_2$ it follows that neither $m \prec C' \prec M_1$ nor $m_2 \prec C' \prec M$ holds.
Therefore we have $w(c_1)+w(c_2)\le w(c)-1$, and $w(\tau_{i+1})<w(\tau_i)$ follows.

We want to prove that if $w(\tau_i) > 0$ then there exist two pairs $c=(m, M)$ and $c'=(m', M')$ in $\tau_i$ such that $c$ and $c'$ are not ordered, thus proving that $w(\tau_i)=0$ holds when we exit the while cycle.
If $w(\tau_i) > 0$, then there exists a pair $c =(m, M)$ in $\tau_i$ such that $w(c)>0$, that is, there exists a Wheeler class $C$ with $m\prec C\prec M$ and $C\nsubseteq[m]_{\equiv_\mt L}$; in particular we have $m \not\equiv^c_{\mathcal L} M$. 
By Lemma \ref{short}, there exists $\alpha \in C$  with $|\alpha|\leq n^2+n$. 
We want to prove, by induction on $j$, that for each $0 \le j \le i$ there exists a pair $c_j =(m_j, M_j) \in \tau_j$ such that $c$ and $c_j$ are not ordered and $m_j\equiv_\mathcal L \alpha$; note that $c$ may not belong to $\tau_j$ for $j < i$.
If $j=0$, let $q_k$ be the state of $\mathcal D_\mt L$ such that $\alpha \in I_{q_k}$.
The pairs $c$ and $c_k = (m_k, M_k)\in \tau_0$ are not ordered, since we have both $m \prec \alpha \prec M$ and $m_k \preceq \alpha \preceq M_k$, therefore we set $c_0 := c_k$.
\\If $0 < j < i$, suppose the thesis holds for $\tau_j$, that is,  there exists a pair $c_j = (m_j, M_j)$ in $\tau_j$ such that $c$ and $c_j$ are not ordered and $m_j \equiv_\mathcal L \alpha $.
If $c_j \in \tau_{i+1}$, we set $c_{j+1}:= c_j$.
Otherwise, $c_j$ had been split into two pairs $c_{j1}=(m_j, M_{j1})$ and $c_{j2}=(m_{j2}, M_{j})$ with $m_j \preceq M_{j1} \prec m_{j2} \preceq M_j$. By construction it holds $m_{j1}\equiv_\mathcal L m_{j2}\equiv_\mathcal L m_j\equiv_\mathcal L\alpha$.
Since $c$ and $c_j$ are not ordered, we have $m \prec M_j$ and $m_j \prec M$.
If $m \prec M_{j1}$ we have both $m \prec M_{j1}$ and $m_j \prec M$, hence $c$ and $c_{j1}$ are not ordered and we set $c_{j+1}:= c_{j1}$; similarly, if $m_{j2} \prec M$ we set $c_{j+1}:= c_{j2}$.
Otherwise we have $M_{j1}\preceq m$ and $M \preceq m_{j2}$ and strings have the following order:
\begin{equation}
\label{imp}
m_j \preceq M_{j1} \preceq m \prec \alpha \prec M \preceq m_{j2} \preceq M_j.
\end{equation}

Let $m'$ be the string used to split $c_j$: by construction it holds $M_{j1} \prec m' \prec m_{j2}$, and it can not be $m'=\alpha$ since $m_j \equiv_{\mathcal L} \alpha$.
By construction, the string $m_{j2}$ is the smallest string $\gamma$ of length at most $n^2+n$ such that $m' \prec \gamma$ and $\gamma \equiv_{\mathcal L} m_j$. Thus, if $m'\prec \alpha$, the string $\alpha$ would have all this properties, hence it would follow that $m_{j2} \preceq \alpha$, which contradicts \eqref{imp}.
Similarly, if $\alpha \prec m'$ it would follow $\alpha \preceq M_{j1}$, a contradiction. 
Therefore the condition depicted in \eqref{imp} can not occur, ending the proof of the inductive step. Hence, if $c\in \tau_i$ and  $w(c)>0$ then there exists $c'\in \tau_i$ such that $c,c'$ are not ordered and we iterate the while cycle. 

Let $\tau_p$ be the last set built before exiting the while cycle.
We need to prove that the collection of the first components of the pairs in Expand$(\tau_p)$ is a fingerprint of $\mt L$.
First we prove that all pairs in $\tau_p$ are ordered. 
Let $c_i=(m_i, M_i)$ and $c_j=(m_j, M_j)$ be two distinct pairs in $\tau_p$.
If $m_i \equiv_\mathcal L m_j$, then $c_i$ and $c_j$ are ordered by invariant \ref{2}.
If instead $m_i \not\equiv_\mathcal L m_j$, suppose $c_i$ and $c_j$ are not ordered. 
Then either $m_i\prec m_j\prec M_i$ or $m_j\prec m_i\prec M_j$. In the first case,    if $C$ is the Wheeler class containing $m_j$,  we have $m_i\prec C \prec M_j$ and $w(c_i)>0$. Similarly,   the second case implies    $w(c_j)>0$ and  in both cases we reach a contradiction with  $w(\tau_p)=0$.

Second, we prove that if $C$ is a Wheeler class that is not represented by $\tau_p$, i.e. $C$ does not contain any components of any pair in $\tau_p$, then there exists a pair $(m, M)\in \tau_p$ such that $m\prec C\prec M$. 
Then the proof is complete, since the subroutine Expand extracts a representative $\alpha_c$ of $C$ and adds to $\tau_p$ the pair $(\alpha_c, \alpha_c)$. 
Let $C$ be a Wheeler class not represented by $\tau_p$, and consider the state $q_i$ in $\mathcal D_\mt L$ such that $C\subseteq I_{q_i}$.
By construction, the pair $(m_{q_i}, M_{q_i})\in \tau_0$ is such that $m_{q_i}$ (respectively, $M_{q_i}$) belongs to the smaller (greatest) Wheeler class contained in $I_{q_i}$.
Since when we build $\tau_{i+1}$ from $\tau_{i}$ we only add, and never delete, representatives of Wheeler classes, both $m_{q_i}$ and $M_{q_i}$ must appear as a component of some pair in $\tau_p$.
Therefore, it is well defined the smallest representative $m$ in $\tau_p$ such that $m\prec C$, as well as the greatest representative $M$ in $\tau_p$ such that $C\prec M$. 
The second part of invariant \ref 2 implies that $m$ and $M$ belong to the same pair, which completes the proof.
\end{proof}

\end{document}